\documentclass[a4paper,11pt]{article}
\usepackage{a4wide}
\usepackage[UKenglish]{babel}
\usepackage[UKenglish]{isodate}
\usepackage{amsmath,amssymb,amsthm,mathtools}
\usepackage[ebgaramond,lf]{newtx}
\usepackage{bm}

\usepackage{etoolbox}
\usepackage{diagbox}
\usepackage{optidef}
\usepackage[svgnames,table]{xcolor}
\usepackage{tensor}
\usepackage{microtype}
\usepackage{thm-restate}
\usepackage{complexity}
\usepackage{comment}
\usepackage{adjustbox}
\usepackage{hyperref} 
\hypersetup{colorlinks=true,citecolor=RoyalBlue,linkcolor=RoyalBlue,urlcolor=RoyalBlue}
\usepackage{tikz}
\usepackage{pgfplots}
\usetikzlibrary{calc}
\usepackage{tikz-cd}
\usepackage{interval}
\intervalconfig{soft open fences}

\definecolor{goldenrod}{RGB}{218, 165, 32}

\DeclareMathOperator{\val}{val}
\DeclareMathOperator{\Inf}{Inf}
\DeclareMathOperator{\NBin}{NBin}

\DeclareMathOperator*{\Exp}{\mathbb{E}}

\newcommand{\vx}{\ensuremath{\mathbf{x}}}
\newcommand{\vy}{\ensuremath{\mathbf{y}}}

\newcommand{\valpha}{\ensuremath{{\alpha}}}
\newcommand{\ve}{\ensuremath{\mathbf{e}}}

\newcommand{\va}{\ensuremath{\mathbf{a}}}

\renewcommand{\S}{\ensuremath{\mathbb{S}}}

\newcommand{\N}{\ensuremath{\mathcal{N}}}

\renewcommand{\R}{\ensuremath{\mathbb{R}}}
\newcommand{\UGC}{{\sf UGC}}

\theoremstyle{plain}
\newtheorem{theorem}{Theorem}
\newtheorem{lemma}[theorem]{Lemma}

\newtheorem*{lemma*}{Lemma}
\newtheorem{proposition}[theorem]{Proposition}
\newtheorem*{proposition*}{Proposition}
\newtheorem{corollary}[theorem]{Corollary}
\newtheorem*{corollary*}{Corollary}
\newtheorem{conjecture}[theorem]{Conjecture}
\theoremstyle{definition}
\newtheorem{definition}[theorem]{Definition}

\pgfplotsset{compat=1.18}

\renewcommand{\epsilon}{\varepsilon}

\begin{document}

\author{Tamio-Vesa Nakajima\\
University of Oxford\\
\texttt{tamio-vesa.nakajima@cs.ox.ac.uk}
\and
Stanislav \v{Z}ivn\'y\\
University of Oxford\\
\texttt{standa.zivny@cs.ox.ac.uk}
}

\title{Maximum $k$- vs. $\ell$-colourings of graphs\thanks{This work was supported by UKRI EP/X024431/1 and by a Clarendon Fund Scholarship. For the purpose of Open Access, the authors have applied a CC BY public copyright licence to any Author Accepted Manuscript version arising from this submission. All data is provided in full in the results section of this paper.}}

\date{\today}
\maketitle

\begin{abstract}

We present polynomial-time SDP-based algorithms for the following problem: For fixed $k \leq \ell$, given a real number $\epsilon>0$ and a graph $G$ that admits a $k$-colouring with a $\rho$-fraction of the edges coloured properly, it returns an $\ell$-colouring of $G$ with an $(\alpha \rho - \epsilon)$-fraction of the edges coloured properly in polynomial time in $G$ and $1 / \epsilon$. Our algorithms are based on the algorithms of Frieze and Jerrum [Algorithmica'97] and of Karger, Motwani and Sudan [JACM'98].

When $k$ is fixed and $\ell$ grows large, our algorithm achieves an approximation ratio of $\alpha = 1 - o(1 / \ell)$.
When $k, \ell$ are both large, our algorithm achieves an approximation ratio of $\alpha = 1 - 1 / \ell + 2 \ln \ell / k \ell - o(\ln \ell / k \ell) - O(1 / k^2)$; if we fix $d = \ell - k$ and allow $k, \ell$ to grow large, this is $\alpha = 1 - 1 / \ell + 2 \ln \ell / k \ell - o(\ln \ell / k \ell)$.

By extending the results of Khot, Kindler, Mossel and O'Donnell [SICOMP'07] to the promise setting, we show that for large $k$ and $\ell$, assuming Khot's Unique Games Conjecture (\UGC), it is \NP-hard to achieve an approximation ratio $\alpha$ greater than $1 - 1 / \ell + 2 \ln \ell / k \ell + o(\ln \ell / k \ell)$, provided that $\ell$ is bounded by a function that is $o(\exp(\sqrt[3]{k}))$. For the case where $d = \ell - k$ is fixed, this bound matches the performance of our algorithm up to $o(\ln \ell / k \ell)$. Furthermore, by extending the results of Guruswami and Sinop [ToC'13] to the promise setting, we prove that it is 
\NP-hard to achieve an approximation ratio greater than $1 - 1 / \ell + 8 \ln \ell / k \ell + o(\ln \ell / k \ell)$, provided again that $\ell$ is bounded as before (but this time without assuming the \UGC).

\end{abstract}

\section{Introduction}

The three most studied objectives in approximation algorithms 
are to maximise the number of satisfied
constraints, to minimise the number of unsatisfied constraints,
and to find a solution that satisfies a
$(1-f(\epsilon))$-fraction of the constraints given an instance in which a
$(1-\epsilon)$-fraction of the constraints is satisfiable, where $f$ is some
function satisfying $f\to 0$ as $\epsilon\to 0$ and not depending on the
input size.\footnote{This notion of tractability, coined \emph{robust
solvability}, was introduced by Zwick~\cite{Zwick98}.}
All three objectives are examples of a
\emph{quantitative} approximation. Another approach to approximation is a
\emph{qualitative} approximation, which insists on satisfying \emph{all}
constraints but possibly in a weaker form. A canonical example of this is the
\emph{approximate graph colouring (AGC)} problem~\cite{GJ76}: Given a $k$-colourable
graph, find an $\ell$-colouring, where $k\leq\ell$.
In this work, we shall combine the two approaches. In particular, we are
interested in the following type of problems: Given a graph in which a large
fraction of edges can be properly $k$-coloured, can we find an $\ell$-colouring
of it with a good fraction of the edges properly coloured? Our main result is an
efficient algorithm for this problem and showing its optimality in many cases.

\medskip
Given a graph $G=(V,E)$ and $k\in\mathbb{N}$, a $k$-colouring of $G$ is an assignment
$c:V\to\{1,\ldots,k\}$ of colours to the vertices of $G$. The value $\rho_k(c)$ of a $k$-colouring  
$c$ is the fraction of properly coloured edges:
\[\rho_k(c)\ =\ \frac{|\{(u,v)\in E\mid c(u)\neq c(v)\}|}{|E|}\,.\]
A $k$-colouring $c$ is called \emph{proper} if $\rho_k(c)=1$; i.e., if no edge is monochromatic under $c$.
We denote by $\rho_k(G)$ 
the largest value of $\rho_k(c)$ over all $k$-colourings $c$ of $G$:
\[\rho_k(G)\ =\ \max_{c:V\to\{1,\ldots,k\}}\rho_k(c)\,.\]
Testing whether $\rho_k(G)=1$ is the same as determining whether $G$ admits a proper
$k$-colouring; this problem is \NP-hard for $k\geq 3$, as shown by
Karp~\cite{Karp1972}, and solvable in
polynomial time for $k=1,2$. Observe also that $\rho_k(G)$ is the (fractional) size of the largest $k$-cut of $G$ --- indeed in the context of maximisation, $k$-cut is the same as $k$-colouring.

Given a graph $G$, say with $\rho=\rho_2(G)<1$ (since the case $\rho_2(G)=1$ is
solvable exactly efficiently), the celebrated result of Goemans and
Williamson uses a semidefinite programming (SDP) relaxation, equivalent to an eigenvalue minimisation problem proposed earlier by Delorme and Poljak~\cite{Delorme93:ejc,Delorme93:mp}, to design a polynomial-time randomised algorithm that finds a $2$-colouring $c$ of $G$ with
$\rho_2(c)\geq \alpha_{\text{GW}}\rho$~\cite{GW95}, where
$\alpha_{\text{GW}}\approx 0.87856$. Their algorithm was later derandomised by Mahajan and Ramesh~\cite{Mahajan99:sicomp}.
On the hardness side, the work of H\r{a}stad~\cite{Hastad01} and Trevisan, Sorkin, Sudan, and
Williamson~\cite{Trevisan00:sicomp} showed that obtaining a $2$-colouring $c$ with
$\rho_2(c)\geq \alpha\rho$ is \NP-hard for any $\alpha\geq 16/17+\epsilon$ for
an arbitrarily small $\epsilon>0$. Note that $16/17\approx 0.94117$ and thus there is a gap between $\alpha_{\text{GW}}$ and $16/17$.
However, under Khot's influential \emph{Unique Games Conjecture} (\UGC)~\cite{Khot02stoc}, Khot, Kindler, Mossel and O'Donnell showed that finding
a $2$-colouring $c$ with $\rho_2(c)\geq \alpha\rho$ is \NP-hard for any
$\alpha\geq\alpha_{\text{GW}}+\epsilon$~\cite{KKMO07},\footnote{The results
in~\cite{KKMO07} was initially conditional on the ``majority is stablest'' conjecture, later proved by Mossel, O’Donnell, and Oleszkiewicz~\cite{Mossel10:ann}.} thus showing that the algorithm of Goemans and Williamson~\cite{GW95} is
optimal (up to an arbitrarily small additive constant).
In fact, the Goemans-Williamson algorithm finds, given a graph $G$ with
$\rho_2(G)=1-\epsilon$, a 2-colouring $c$ of $G$ with
$\rho_2(c)=1-O(\sqrt{\epsilon})$~\cite{GW95}. Moreover, the dependence on
$\epsilon$ is \UGC-optimal~\cite{KKMO07}.

What about colourings with more than two colours? Building on the work of
Goemans and Williamson~\cite{GW95}, Frieze and Jerrum~\cite{FJ97} provided
an SDP-based algorithm for approximating $\rho_k(G)$ for every $G$ and constant $k\geq 2$.
Asymptotic optimality of this algorithm for large $k$ (up to an arbitrarily small additive constant) was shown by Khot,
Kindler, Mossel, and O'Donnell under the \UGC~\cite{KKMO07}, as we will discuss
in more detail later.
All the results mentioned so far are
concerned with quantitative approximation.
We now turn to qualitative approximation. 

Let $G$ be a graph that can be properly $k$-coloured; i.e., $\rho_k(G)=1$. Is it possible to find efficiently a proper $\ell$-colouring of $G$ for some constant
$k\leq\ell$? Garey and Johnson conjectured that this problem is \NP-hard as long
as $k\geq 3$~\cite{GJ76}. For $k=3$, \NP-hardness is known for
$\ell=3$~\cite{Karp1972}, $\ell=4$~\cite{KLS00,GK04}, and
$\ell=5$~\cite{BBKO21}; the case of $\ell\geq 6$ is open. For $k\geq 4$, better bounds are known~\cite{KOWZ23}. However, \NP-hardness has
been shown for all constant $3\leq k\leq\ell$ under stronger assumptions.
Namely, under a non-standard variant of the Unique Games Conjecture
by Dinur, Mossel, and Regev~\cite{Dinur09:sicomp},
under the $d$-to-1 conjecture of Khot~\cite{Khot02stoc} (for any fixed $d$)
by Guruswami and Sandeep~\cite{GS20:icalp}, and
under the rich $2$-to-$1$ conjecture of Braverman, Khot, and Minzer~\cite{Braverman21:itcs}
by Braverman, Khot, Lifshitz, and Minzer in~\cite{Braverman21:focs}.

\medskip
We now combine the quantitative and qualitative approaches. 
Given a graph $G$ of value $\rho_k(G)$, what is the largest
$0<\alpha\leq 1$ so that $\rho_\ell(G)$ can be $\alpha$-approximated?
It is not hard to show 
that, for $3\leq k\leq\ell$, a $1$-approximation is at least as hard as approximate graph colouring, cf. Section~\ref{app:alpha1}.
For $k = 2, \ell = 3$, one can get an approximation ratio of 1, via a deterministic algorithm~\cite{nz24:arxiv-bipartite}.
For $\alpha < 1$, not much is known other than what follows immediately from the already mentioned previous work: The algorithm from~\cite{FJ97} gives an $\alpha$-approximation with $\alpha\geq 1-1/\ell+(1+\epsilon(\ell))(2\ln \ell/\ell^2)$, 
where $\epsilon(\ell)\to 0$ as $\ell\to\infty$,
and for $\ell=k\geq 2$ this algorithm is \UGC-optimal (up to an arbitrarily small additive constant)~\cite{KKMO07}.
However, the situation is unclear for general $k$ and $\ell$.

\paragraph{Contributions}
We initiate a  systematic investigation of promise maximum colouring, i.e., $k$- vs. $\ell$-colourings.
As our first result, we extend the algorithm of Frieze and Jerrum~\cite{FJ97} to
work for $k$- vs. $\ell$-colourings. We analyse the power of the algorithm 
for $k \leq \ell$ as $k, \ell \to \infty$.
\begin{restatable}{theorem}{main}\label{thm:main}
Fix $2 \leq k \leq \ell$. There is a randomised algorithm which, given a graph $G$
  that admits a $k$-colouring of value $\rho$ and a real number $\epsilon>0$, finds an $\ell$-colouring of expected
  value $\alpha_{k\ell}\rho - \epsilon$ in polynomial time in $G$ and $\log (1 / \epsilon)$. In particular,
  \begin{enumerate}
      \item\label{thm:main:bound1} We have that
      \[ \alpha_{k\ell} \geq 1 - \frac{1}{\ell} + \frac{2 \ln \ell}{k \ell} - o\left(\frac{\ln \ell}{k\ell}\right) - O\left(\frac{1}{k^2}\right).\]
      \item\label{thm:main:bound2} Moreover, we have that $\alpha_{k\ell} > 1 - 1 / \ell$,
      hence the algorithm is better than random guessing for $\rho$ near 1.\footnote{The $o(\cdot), O(\cdot)$ notation hides only terms and factors dependant on $k, \ell$, not on $\epsilon$.}
  \end{enumerate}
\end{restatable}
Observe that point 1 above does not imply that $\alpha_{k\ell} \geq 1$. Indeed, suppose that $\eta$ is the constant hidden in $O(1 / k^2)$ and that $\ell$ is large enough so that $\sqrt{\eta / \ell} > \ln \ell / \ell$. Then by the AM-GM inequality,
\[ 
1 - \frac{1}{\ell} + \frac{2 \ln \ell}{k\ell} - \frac{\eta}{k^2}
\leq 1 - \frac{2\sqrt{\eta}}{k \sqrt \ell} + \frac{2\ln \ell}{k \ell} < 1.
\]
{
\renewcommand{\arraystretch}{1.3}
\begin{table}
    \centering
    \rowcolors{2}{gray!15}{white}
    {
    \footnotesize
    \begin{tabular}{|c|ccccccccccccc|}
    \hline
         \backslashbox{$k$}{$\ell$} & 3 & 4 & 5 & 6 & 7 & 8 & 9 & 10 & 11 & 12 & 13 & 14 & 15 \\
         \hline
         3  &  .836  &  .904  &  .938  &  .957  &  .969  &  .976  &  .982  &  .985  &  .988  &  .990  &  .992  &  .993  &  .994  \\
4  & &  .858  &  .899  &  .924  &  .940  &  .952  &  .960  &  .967  &  .972  &  .975  &  .979  &  .981  &  .983  \\
5  & & &  .877  &  .904  &  .923  &  .936  &  .946  &  .954  &  .960  &  .964  &  .968  &  .972  &  .974  \\
6  & & & &  .892  &  .911  &  .926  &  .936  &  .945  &  .952  &  .957  &  .961  &  .965  &  .968  \\
7  & & & & &  .903  &  .918  &  .930  &  .938  &  .945  &  .951  &  .956  &  .960  &  .963  \\
8  & & & & & &  .913  &  .924  &  .934  &  .941  &  .947  &  .952  &  .956  &  .960  \\
9  & & & & & & &  .920  &  .930  &  .937  &  .944  &  .949  &  .953  &  .957  \\
10  & & & & & & & &  .927  &  .935  &  .941  &  .946  &  .951  &  .954  \\
11  & & & & & & & & &  .932  &  .939  &  .944  &  .949  &  .953  \\
12  & & & & & & & & & &  .937  &  .942  &  .947  &  .951  \\
13  & & & & & & & & & & &  .941  &  .946  &  .950  \\
14  & & & & & & & & & & & &  .944  &  .949  \\
15  & & & & & & & & & & & & &  .948  \\\hline
    \end{tabular}
    }
    \caption{Approximate values of $\alpha_{k\ell}$. As all values are between 0 and 1, we omit the leading 0.}
    \label{tab:table}
\end{table}
}
For illustration, we tabulated numerical approximate values for $\alpha_{k\ell}$ in Table~\ref{tab:table}.\footnote{The exact definition of $\alpha_{k\ell}$ is given in Definition~\ref{def:alphakl}. The probabilities $P_\ell(a)$ that appear in that definition were computed using the methods and R library from~\cite{AG18}.}
Our algorithm solves the Frieze-Jerrum SDP for $k$-colourings, then rounds like Frieze and Jerrum do for $\ell$-colourings~\cite{FJ97}. 
We largely follow the analysis from~\cite{FJ97}.

We will also show how to derandomise our algorithm.

\begin{restatable}{theorem}{maindet}\label{thm:maindet}
    Fix $2 \leq k \leq \ell$ and let $\alpha_{k\ell}$ be as in Theorem~\ref{thm:main}.
    There is a deterministic algorithm which, given a graph $G$ that admits a
    $k$-colouring of value $\rho$
    and a real number $\epsilon
    >0$, finds an $\ell$-colouring of value $\alpha_{k\ell} \rho - \epsilon$ in polynomial time in $G$ and $1 / \epsilon$.
\end{restatable}
The proof of Theorem~\ref{thm:maindet} builds on a general derandomisation theorem from~\cite{nz24:arxiv-bipartite}, which in turn 
uses the method of conditional expectation.

The algorithm from Theorem~\ref{thm:maindet} has good performance when both $k$ and $\ell$ grow large. What if $k$ is fixed and only $\ell$ grows large? We give an algorithm that has  good performance  in this case as well.
The idea is based on an algorithm of Karger, Motwani and Sudan for approximate graph colouring~\cite[Section 6]{Karger98:jacm}, but rather than cutting by $\Theta(\log(n))$ random hyperplanes we cut with $\lfloor \log_2(\ell) \rfloor$ hyperplanes.

\begin{restatable}{theorem}{thmbigk}\label{thm:bigk}
Let $k > 2$ be fixed and $\ell\geq k$ be large. There is a deterministic algorithm which,  given a graph $G$ that admits a $k$-colouring  of value $\rho$ and a real number $\epsilon > 0$,  finds an $\ell$-colouring of $G$ of value $\alpha_{k\ell}' \rho - \epsilon$ in polynomial time in $G$ and $1 / \epsilon$. In particular, for a fixed $k$ there exists a constant $u_k > 1$ such that
\[
\alpha_{k\ell}' \geq 1 - O\left(1/ \ell^{u_k}\right),
\]
which is 
$1 - o(1 / \ell)$
when $\ell$ grows large.\footnote{For example, for $k = 3$ we have $u_k \approx 1.58$ and the approximation ratio is approximately $1 - O(1 / \ell^{1.58})$. This is significantly better, for large $\ell$, than the random guessing algorithm which has performance $1 - 1 /\ell$.}
\end{restatable}
%
Of course, by running the algorithms of Theorem~\ref{thm:maindet} and Theorem~\ref{thm:bigk} in parallel and then taking the better of the two results, we can get an algorithm that is at least as good as either of them.

We now turn to hardness results.
Using the framework of Khot, Kindler, Mossel, and O'Donnel~\cite{KKMO07}, we
  will show that,
  under the \UGC,
  it is \NP-hard to beat the approximation guarantee of our algorithm from Theorem~\ref{thm:maindet} by more than a constant that grows small for any large $k, \ell$ with $k\leq\ell$ and $\ell$ bounded by a function that is $o(e^{\sqrt[3]{k}})$. We combine this with the methods of Guruswami and Sinop~\cite{Guruswami13:toc}\footnote{The results of~\cite{Guruswami13:toc} were, at the time, conditional on the 2-to-1 conjecture of Khot~\cite{Khot02stoc}; however this has recently been proved by Khot, Minzer and Safra~\cite{Khot23:annals}.}
  to also find some weaker (non-tight) unconditional results.\footnote{We thank Venkat Guruswami for bringing~\cite{Guruswami13:toc} to our attention.} We will present a unified version of these proofs, using ideas from the work of Dinur, Mossel and Regev~\cite{Dinur06:STOC} (on which~\cite{Guruswami13:toc} also draws). The way the unification of these two proofs works out is also similar to the work of Guruswami and Sandeep~\cite{GS20:icalp}.
\begin{restatable}{theorem}{hardness}\label{thm:hardness}
    Fix some function $M(k) = o(e^{\sqrt[3]{k}})$.
    Let $2 \leq k \leq \ell$ be such that $\ell \leq M(k)$. For any small enough $\epsilon > 0$, consider the problem of deciding whether a given graph $G$ admits a $k$-colouring of value $1 - \epsilon$, or not even an $\ell$-colouring of value $\beta + \epsilon$. We have the following.
    \begin{itemize}
        \item Assuming the \UGC, the problem is \NP-hard for \[\beta = \beta_{k\ell} = 1 - \frac{1}{\ell} + \frac{2 \ln \ell}{k \ell} + o\left(\frac{\ln \ell}{k \ell}\right).\]
        \item Unconditionally, the problem is \NP-hard for \[\beta = \beta_{k\ell}'= 1 - \frac{1}{\ell} + \frac{8 \ln \ell}{k \ell} + o\left(\frac{\ln \ell}{k \ell}\right).\]
    \end{itemize}
    Both of these results only hold when $\beta_{k\ell}, \beta'_{k\ell} \in (0, 1)$.\footnote{The constants hidden in the
    expression defining $\beta_{k\ell}$ depend on $M(k), k, \ell$, but not on $\epsilon$.}
\end{restatable}
The \NP-hardness bound in Theorem~\ref{thm:hardness} is limited due to the fact that, for a fixed $k$, we cannot have $\ell$ arbitrarily large. This is intrinsic to the expression above: for a large $\ell$ and a fixed $k$ we have $\beta_{k\ell} > 1$. Moreover, any \NP-hardness bound for a fixed $k$ and a large $\ell$ must take into account the algorithm with approximation ratio $1 - o(1 / \ell)$ we gave in Theorem~\ref{thm:bigk}.
The proof of Theorem~\ref{thm:hardness} can be found in Section~\ref{sec:hardness}.

Finally, we present a simple reduction from approximate graph colouring with
perfect completeness (i.e.~fix $k \leq \ell$; then given a graph $G$, output yes
if $G$ is $k$-colourable, and no if $G$ is not even $\ell$-colourable) to
1-approximating $k$- vs.~$\ell$-colouring. A simple proof of the following can be found in Section~\ref{app:alpha1}.

\begin{restatable}{proposition}{alphaone}\label{prop:alpha1}
    Fix $3\leq k\leq\ell$ and some rational $\rho \in \interval[open left]{0}{1}$. There is a
    log-space reduction from the problem of distinguishing $\rho_k(G)=1$ vs.
    $\rho_\ell(G)<1$ to the problem of distinguishing $\rho_k(G)\geq \rho$ vs.
    $\rho_\ell(G) < \rho$.
\end{restatable}

\paragraph{Evaluation of performance bounds} 
For $k = \ell$ we recover
the positive results of~\cite{FJ97} (and indeed our algorithm is the same as
that of~\cite{FJ97} for $k = \ell$) and the negative result of~\cite{KKMO07}.
In detail, we have that our algorithm from Theorem~\ref{thm:main} has performance $1 - 1 / \ell + 2 \ln \ell /
\ell^2 -o(\ln \ell/\ell^2)$ and that it is \NP-hard to do, under the \UGC, any
better than $1 - 1 / \ell + 2 \ln \ell / k\ell + o(\ln \ell /
\ell^2)$.\footnote{The negative result of~\cite{KKMO07} is slightly more
specific as their asymptotic error term is $O(\ln{\ln \ell  }/ \ell^2)$. A careful
inspection of our analysis shows that our error term is $O(1 / \ell^2 + \ln \ell \ln \ln k / \ell k \ln k)$, which for $\ell = k$ is precisely their $O( \ln \ln \ell / \ell^2)$.}
We also recover the unconditional (in light of~\cite{Khot23:annals}) result of~\cite{Guruswami13:toc}, i.e.~that it is \NP-hard to do any better than $1 - 1 / \ell + 8 \ln \ell / k\ell + o(\ln \ell / \ell^2)$.

For the fixed-gap case, i.e.~$\ell = k + d$ for some fixed $d \geq 0$, we get
the same type of result as for $k=\ell$: performance $1 - 1 / \ell + 2 \ln \ell /
k\ell - o( \ln \ell / k \ell)$ and \NP-hardness, under the \UGC, of $1 - 1 / \ell + 2 \ln \ell
/ k\ell + o(\ln \ell / k \ell)$, since in this case the $1 / k^2$ term is
strictly dominated by $\ln \ell / k \ell$.\footnote{As an example, this implies
that for large $k$ we can do $k$- vs. $(\ell = k+10)$-colourings with approximation ratio $1 - 1 / \ell + 1.999 \ln \ell / k\ell$, but it is \NP-hard under the \UGC~to do it with approximation ratio $1 - 1 / \ell + 2.001 \ln \ell / k\ell$.} Furthermore, unconditionally we find that an approximation ratio of $1 - 1 / \ell + 8 \ln \ell / k \ell + o(\ln \ell / k \ell)$ is \NP-hard to achieve. This unconditional result is not yet tight, since already the second order term is different.

For fixed $k$ and large $\ell$, the algorithm from Theorem~\ref{thm:bigk} has performance $1 - o(1 / \ell)$. This algorithm cannot be improved by more than $o(1 / \ell)$, since no algorithm can have approximation ratio greater than 1. We believe that the algorithm from Theorem~\ref{thm:maindet} is at least as strong as the algorithm from Theorem~\ref{thm:bigk} even for fixed $k$ and large $\ell$.

We note that the fact that our \NP-hardness bound in Theorem~\ref{thm:hardness} only works for $\ell$ bounded by some function of $k$ mirrors the current state-of-the-art for approximate graph colouring:  distinguishing proper $k$- vs. $\ell$-colourings is only known to be \NP-hard whenever $\ell \leq \binom{k}{\lfloor k / 2\rfloor} - 1$~\cite{KOWZ23}.

\paragraph{Related work}
Graph colouring is a canonical example a \emph{Constraint Satisfaction Problem}
(CSP)~\cite{Feder98:monotone,KSTW00}.
Robust solvability of CSPs was studied, among others, by Charikar, Makarychev,
and Makarychev~\cite{Charikar09:talg}, Guruswami and Zhou~\cite{GZ12:toc}, 
and Barto and Kozik~\cite{BK16:sicomp}. 
Raghavandra showed \UGC-optimality of the basic SDP 
programming relaxation for all CSPs~\cite{Raghavendra08:everycsp}.
The notion of an almost $k$-colouring (a large fraction of the graph being properly $k$-coloured) was recently studied by Hecht, Minzer, and Safra~\cite{Hecht23:approx}, who showed that finding an almost $k$-colouring of a graph that admits an almost $3$-colouring is \NP-hard for every constant $k$.
Austrin, O'Donnell, Tan and Wright showed \NP-hardness of distinguishing whether $\rho_3(G)=1$ or $\rho_3(G)<\frac{16}{17}+\epsilon$~\cite{Austrin14:toct}.

Approximate graph colouring is an example of  a \emph{Promise Constraint Satisfaction Problem} (PCSP)~\cite{AGH17,BG21:sicomp,BBKO21}.
Robust solvability of PCSPs has recently been investigated by
Brakensiek, Guruswami, and Sandeep~\cite{BGS23:stoc}.
Bhangale, Khot, and Minzer have recently studied approximability of certain Boolean PCSPs~\cite{Bhangale22:stoc,Bhangale23:stoc2,Bhangale23:stoc3}.

\section{Preliminaries}

For any positive integer $n$ let $[n] = \{ 1, \ldots, n \}$.
For any predicate $\phi$, we let $[\phi] = 1$ if $\phi$ is true, and $0$ otherwise.
We shall use semidefinite programming and refer the reader to~\cite{Gartner2012approximation} for a reference.

For an event $\phi$ we let $\Pr [\phi]$ be the probability that $\phi$ is true. For a random variable $X$, we let $\Exp[X]$ denote its expected value. Note that $\Exp[ [\phi]] = \Pr[\phi]$.

For any two distributions $\mathcal{D}, \mathcal{D}'$ with domains $A, A'$, we let $\mathcal{D} \times \mathcal{D}'$ denote the product distribution, whose domain is $A \times A'$.
For any distribution $\mathcal{D}$ over $\R$ and $a, b \in \R$, the distribution $a\mathcal{D} + b$ is the distribution of $aX + b$ when $X \sim \mathcal{D}$.
We use the standard probability theory abbreviations i.i.d.~(independent and identically distributed) and p.m.f.~(probability mass function).

We introduce a few classic distributions we will need. The uniform distribution $\mathcal{U}(D)$ over a finite set $D$ is the distribution with p.m.f.~$f : D \to [0, 1]$ given by $f(x) = 1 / |D|$. Note that $\mathcal{U}(D^n)$ is the same as ${\mathcal{U}(D)}^n$, a fact which we will use implicitly. We let $\NBin(n)$ denote a normalised binomial distribution: it is the distribution of $X_1 + \cdots + X_n$, where $X_i \sim \mathcal{U}(\{-1/\sqrt{n}, 1/\sqrt{n}\})$. The domain of this distribution is $\{(-n + 2k) / \sqrt{n} \mid 0 \leq k \leq n\}$, the probability mass function is $(-n + 2k)/\sqrt{n} \mapsto \binom{n}{k} / 2^n$, the expectation is 0, and the variance is 1.
If $\mu, \sigma \in \R$, then we let $\mathcal{N}(\mu, \sigma^2)$ denote the normal distribution with mean $\mu$ and variance $\sigma^2$. Fixing $d$, if $\mathbf{\mu} \in \R^d, \mathbf{\Sigma} \in \R^{d\times d}$, then we let $\mathcal{N}(\mathbf{\mu}, \mathbf{\Sigma})$ denote the multivariate normal distribution with mean $\mathbf{\mu}$ and covariance matrix $\mathbf{\Sigma}$. We let $\mathbf{I}_d$ denote the $d \times d$ identity matrix. Observe that if $\mathbf{x} \sim \mathcal{N}( \mathbf{\mu}, \mathbf{\Sigma})$, where $\mathbf{x} \in \R^d$, then for any matrix $\mathbf{A} \in \R^{d' \times d}$ we have that $\mathbf{A} \mathbf{x} \sim \mathcal{N}(\mathbf{A} \mathbf{\mu}, \mathbf{A} \mathbf{\Sigma} \mathbf{A}^T)$. Furthermore if $\mathbf{x} \sim \mathcal{N}(\mathbf{\mu}, \mathbf{\Sigma})$ with $\mathbf{\Sigma}$ positive semidefinite, then by finding the Cholesky decomposition $\mathbf{\Sigma} = \mathbf{A} \mathbf{A}^T$, where $\mathbf{A} \in \R^{d \times d}$, we find that $\vx$ is identically distributed to $\mathbf{A} \vx' + \mathbf{\mu}$, where $\vx' \sim \mathcal{N}(\mathbf{0}, \mathbf{I}_d)$.

\section{Main result}\label{sec:main}

In this section, we will prove our main result, restated here.

\main*

\noindent
In order to prove Theorem~\ref{thm:main}, we first introduce an auxiliary notion, which already appears in~\cite{FJ97}.

\begin{definition}
    Fix $a, b \in \mathbb{R}$ such that $a^2 + b^2 = 1$, $b \geq 0$, and $\ell \in
    \mathbb{N}$. Suppose that $x_1, \ldots, x_\ell, y_1, \ldots, y_\ell \sim \mathcal{N}(0, 1)$; i.e., they are i.i.d.~standard normal variables. We let $P_\ell(a)$ denote the probability that 
    \[
    \begin{matrix}
    & x_1 \geq x_2 & \wedge & \cdots & \wedge & x_1 \geq x_\ell \\
    \wedge & a x_1 + b y_1 \geq  a x_2 + b y_2 & \wedge & \cdots & \wedge & a x_1 + b y_1 \geq  a x_\ell + b y_\ell.
    \end{matrix}
    \]
    We then write $N_\ell(a) = \ell P_\ell(a)$. This is just the probability that
    \[
    \arg \max_c x_c = \arg \max_c (ax_c + b y_c).
    \]
\end{definition}
\noindent
The following quantity is similar to $\alpha_k$ from~\cite{FJ97}.

\begin{definition}\label{def:alphakl}
    Let
    \[
    \alpha_{k\ell} = \min_{-1 / (k-1) \leq a < 1} \frac{k(1 - \ell P_\ell(a))}{(k-1)(1 - a)}.
    \]
    Observe that for $a = 1$ the ratio would be $0 / 0$, hence for $-1/(k-1) \leq a \leq 1$ it holds that
    \begin{equation}\label{ineq:akl}
    \alpha_{k\ell} \frac{k-1}{k}(1-a) \leq 1 - \ell P_\ell(a).
    \end{equation}
\end{definition}
\noindent
The proof of Theorem~\ref{thm:main} is split into the following three propositions.

\begin{restatable}{proposition}{propalgo}\label{prop:algo}
There is a randomised algorithm which, given a graph $G$
  that admits a $k$-colouring of value $\rho$, finds an $\ell$-colouring of expected
  value $\alpha_{k\ell}\rho - \epsilon$ in polynomial time in $G$ and $\log (1 / \epsilon)$ for an arbitrarily small $\epsilon>0$.
\end{restatable}

\begin{restatable}{proposition}{propboundskl}\label{prop:boundskl}
  $\displaystyle \alpha_{k\ell} \geq 1 - \frac{1}{\ell} + \frac{2 \ln \ell}{k \ell} - o\left(\frac{\ln \ell}{k\ell} \right) - O\left(\frac{1}{k^2}\right)$.
\end{restatable}

\begin{restatable}{proposition}{propnontrivial}\label{prop:nontrivial}
$\alpha_{k\ell} > 1 - 1 / \ell$.
\end{restatable}

\subsection{Proof of Proposition~\ref{prop:algo}}

Our algorithm solves the SDP of~\cite{FJ97} for $k$-colourings, then rounds the solution of the SDP like~\cite{FJ97} but for $\ell$-colourings. Henceforth fix $2 \leq k \leq \ell$, and $\epsilon > 0$. The following lemma also appears, essentially, as~\cite[Lemma 3]{FJ97} and the preceding definitions; we include it for completeness.

\begin{lemma}
    For any $n \geq k$, there exist vectors $\ve_1, \ldots, \ve_k \in \mathbb{R}^n$ such that $\ve_i \cdot \ve_i = 1$ and $\ve_i \cdot \ve_j = -1 / (k - 1)$ for $i \neq j$.
\end{lemma}
\begin{proof}
    We take
    \[
    \ve_i = \frac{1}{\sqrt{k(k-1)}}{(1, \ldots, 1, 1-k, 1, \ldots, 1, 0, \ldots, 0)}^T,
    \]
    where there are $k$ nonzero values, and where the value $1-k$ appears at the
    $i$-th position. These vectors satisfy the required conditions.
\end{proof}

\begin{proof}[Proof of Proposition~\ref{prop:algo}]
Suppose we are given a graph $G = (V, E)$, which we are promised has a
  $k$-colouring of value $\rho$. Suppose $V = [n]$ and $|E| = m$. If $n \leq k$, then assigning each vertex a different colour satisfies all edges, so assume $n \geq k$.
  
By relabelling the promised colouring to $\ve_1, \ldots, \ve_k \in
  \mathbb{R}^n$, we find that there exist variables $\va_i^* \in \{\ve_1,
  \ldots, \ve_k\} \subseteq \mathbb{R}^n$ for $i \in [n]$ such that
  \[
  \frac{1}{m}
  \sum_{(i, j) \in E} [\va_i^* \neq \va_j^*] \geq \rho.
  \]
  We find that $[ \va_i^* \neq \va_j^* ] =
  \frac{k-1}{k}(1 - \va_i^* \cdot \va_j^*)$, when $\va_i^*, \va_j^* \in \{\ve_1,
  \ldots, \ve_k\}$. We now relax as Frieze and Jerrum~\cite{FJ97}, and Goemans
  and Williamson before them~\cite{GW95}, to a semidefinite program; namely, we solve the following program:
\begin{maxi}|s|
{}{\frac{1}{m} \sum_{(i, j) \in E} \frac{k - 1}{k}(1 - \va_i \cdot \va_j)} 
{}{}\label{sdp:gw}
\addConstraint{\va_i \cdot \va_i = 1}{}
\addConstraint{\va_i \cdot \va_j \geq -\frac{1}{k-1}, i \neq j} 
\addConstraint{\va_i \in \mathbb{R}^n}. 
\end{maxi}
The semidefinite program~(\ref{sdp:gw}) can be solved with an additive error of at most $\epsilon / \alpha_{k\ell}$ in time polynomial with respect to $n, m$ and $\log(\alpha_{k\ell}/\epsilon) = \log(1 / \epsilon) + O(1)$.
By the discussion in the previous paragraph, we see that the SDP must have value at least $\rho$, due to the potential solution $\va_1^*, \ldots, \va_n^*$.
Thus, by solving the program we now have a collection of $n$ unit vectors $\va_1, \ldots, \va_n \in \mathbb{R}^n$ with pairwise inner product at least $-1 / (k - 1)$ such that
\[
\frac{1}{m} \sum_{(i, j) \in E}\frac{k - 1}{k}(1 - \va_i \cdot \va_j) \geq \rho - \epsilon / \alpha_{k\ell}.
\]

Our algorithm now \emph{randomly rounds} as Frieze and Jerrum does~\cite{FJ97}, for $\ell$-colourings. Namely, we take $\ell$ standard normal variables $\vx_1,\ldots , \vx_\ell \in \mathbb{R}^n$; for each vertex $i \in V$ we compute $c = \arg \max_j \va_i \cdot \vx_j$, and then assign vertex $i$ colour $c$ (breaking possible ties arbitrarily).

Now, let us compute the expected value of the resulting rounding. Consider an edge $(i, j) \in E$; in terms of $\frac{k - 1}{k}(1 - \va_i \cdot \va_j)$, what is the probability that $(i, j)$ is properly coloured? This is the same as the probability that $\arg \max_c \va_i \cdot \vx_c \neq \arg \max_c \va_j \cdot \vx_c$, which, by symmetry, is equal to
\begin{equation}\label{eq:prob}
1 - \ell \Pr_{\mathbf{x}_1, \ldots, \mathbf{x}_\ell}\left[
\bigwedge_{c =2}^\ell \va_i \cdot \vx_1 \geq \va_i \cdot \vx_c,
\bigwedge_{c =2}^\ell \va_j \cdot \vx_1 \geq \va_j \cdot \vx_c \right].
\end{equation}
Since $\vx_1, \ldots, \vx_\ell$ are drawn from a rotationally symmetric distribution, we can rotate everything to be in a 2-dimensional plane without affecting the probability in~(\ref{eq:prob}). Furthermore, rotate so that $\va_i$ is moved to $(1, 0)$, and $\va_j$ is at $(a, b)$, where $a = \va_i \cdot \va_j$ and $b = \sqrt{1 - a^2}$ (note that this rotation is possible since it preserves the angle between $\va_i$ and $\va_j$, and their lengths). Since the vectors $\vx_1, \ldots, \vx_\ell$ are (after the rotation) bivariate standard normal variables, we can see them as pairs $(x_1, y_1), \ldots, (x_\ell, y_\ell)$, where $x_1, \ldots, x_\ell, y_1, \ldots, y_\ell \sim \mathcal{N}(0, 1)$ are i.i.d.~standard normal variables. Then, we can rewrite~\eqref{eq:prob} as
\begin{multline}\label{ineq:final}
1 - \ell \Pr_{\substack{x_1, \ldots, x_\ell \\ y_1, \ldots, y_\ell}}\left[\bigwedge_{c = 2}^\ell x_1 \geq x_c, \bigwedge_{c = 2}^\ell a x_1 + b y_1 \geq a x_c + b y_c \right] \\
= 1 - \ell P_\ell(a) = 1 - \ell P_\ell(\va_i \cdot \va_j).
\end{multline}
Since $-1/(k-1) \leq \va_i \cdot \va_j \leq 1$, by (\ref{ineq:akl}), we have that
\[
\alpha_{k\ell} \frac{k - 1}{k}(1 - \va_i \cdot \va_j) \leq 1 - \ell P_\ell(\va_i \cdot \va_j).
\]
Hence, by linearity of expectation the expected value of the $\ell$-colouring we return is, as required,  at least
\[
\alpha_{k\ell} \frac{1}{m} \sum_{(i, j) \in E} \frac{k - 1}{k}(1 - \va_i \cdot \va_j) \geq \alpha_{k\ell} \rho - \epsilon.\qedhere
\]
\end{proof}

\subsection{Proof of Proposition~\ref{prop:boundskl}}

Our proof of Proposition~\ref{prop:boundskl}, restated below, very closely follows~\cite[Corollary 6, Corollary 7]{FJ97}.

\propboundskl*
\noindent
The following result follows from the analysis in~\cite[Lemma 5, Corollary 6, Corollary 7]{FJ97}.

\begin{theorem}
    The Taylor series for $N_\ell(x)$, given by
    \[
    N_\ell(x) = \sum_{i = 0}^\infty c_i x^i
    \]
    converges for $-1 \leq x \leq 1$. Every $c_i \geq 0$. Furthermore $c_0 = 1 / \ell$, $c_1 \sim 2 \ln \ell / (\ell - 1)$, and $\sum_{i= 0}^\infty c_{2i} = 1/2$.\footnote{Some of these facts do not appear in the statements of~\cite[Lemma 5, Corollary 6, Corollary 7]{FJ97}, only in the proofs.}
\end{theorem}
\noindent
The following fact was observed in~\cite{FJ97}; we include a proof for completeness.
\begin{lemma}\label{lem:nonengativeBound}
For $0 \leq a \leq 1$, we have $\displaystyle \frac{k-1}{k} (1-a) \leq 1 - N_\ell(a)$.
\end{lemma}
\begin{proof}
    We first wish to find $N_\ell(0), N_\ell(1)$.
    $P_\ell(0)$ is just the probability that $x_1 \geq x_i$ and $y_1 \geq y_i$ for $x_1, \ldots, x_\ell, y_1, \ldots, y_\ell \sim \mathcal{N}(0, 1)$. By symmetry these events occur with probability $1 / \ell$ each, and thus overall they occur with probability $1 / \ell^2$. On the other hand, $P_\ell(1)$ is just the probability that $x_1 \geq x_i$ for $x_1 ,\ldots, x_\ell \sim \mathcal{N}(0, 1)$. By symmetry this is $1 / \ell$.

    Observe that since every term in the Taylor series of $N_\ell$ is
    nonnegative, the function is convex on $[0, 1]$, hence $1 - N_\ell(a)$ is
    concave. Since furthermore $(k - 1)(1 - 0) / k  = 1 - 1/k \leq 1 - 1 / \ell
    = 1 - N_\ell(0)$ and $(k - 1)(1 - 1) / k = 0 \leq 0 = 1 - N_\ell(1)$, by
    Jensen's inequality we have, for $0 < a < 1$, that
    \[
    \frac{k - 1}{k}(1 - a) < 1 - \ell N_\ell(a).\qedhere
    \]
\end{proof}

\begin{proof}[Proof of Proposition~\ref{prop:boundskl}]
First observe that we only need to prove this result for large enough $k$; for all small $k$ we can just force the bound to hold by increasing the $o(\cdot)$ term arbitrarily. Thus we will prove that the bound holds only for large enough $k$.
    We will try to find some $R \in [1/2, 1]$ such that
    \begin{equation}\label{requirement}
    R \frac{k-1}{k}(1-a) \leq 1 - \ell P_\ell(a) = 1 - N_\ell(a),
    \end{equation}
    for $-1/(k-1) \leq a \leq 1$. We will then conclude $\alpha_{k\ell} \geq R$. (The stipulation that $R \geq 1 / 2$ will be necessary later; it is justified by the fact that at the end we will find such an $R$.)

    By Lemma~\ref{lem:nonengativeBound},~\eqref{requirement} is true for any $0 \leq R \leq 1$ and $0 \leq a \leq 1$. In other words, we need only to care about $-1/(k-1) \leq a \leq 0$; thus assume that this is the case.

    Now, for $-1/(k-1) \leq a \leq 0$, we have $a^{2i} \leq a^2$ and $a^{2i +
    1} \leq 0$; since we know the first two coefficients of the Taylor series of
    $N_\ell$, and the sum of the even coefficients, by ignoring the higher-order
    odd terms and summing together the even terms we can deduce therefore
    that
    \begin{equation}\label{eq1}
    N_\ell(a) \leq \frac{1}{\ell} + (1 + \epsilon(\ell)) \frac{2 \ln \ell}{\ell} a + \frac{a^2}{2},
    \end{equation}
    where $\lim_{\ell \to \infty} \epsilon(\ell)  = 0$. (This is because we know that the first-order coefficient is of order $2 \ln \ell / (\ell - 1) \sim 2 \ln \ell / \ell$.)
    We suppress the $\ell$ in $\epsilon(\ell)$ henceforth.
    
By substituting~\eqref{eq1} into~\eqref{requirement} and factoring out $(-a)$, we get the following sufficient condition on $R$
\[
1 - \frac{1}{\ell} + (-a) \underbrace{\left((1 + \epsilon) \frac{2\ln \ell}{\ell} + \frac{a}{2}\right)}_{A} \geq \frac{k-1}{k}(1 - a) R.
\]
Observe that the left-hand side is a linear function of $A$ with nonnegative slope;
  thus by substituting $A$ with its minimum value we get another sufficient
  condition on $R$. Observe that the value of $a$ that minimises $A$ is $a = -1
  / (k-1)$, i.e.~the minimum value. Hence the following holding for all $-1/(k-1) \leq a \leq 0$ is a sufficient condition for $R$:
\[
1 - \frac{1}{\ell} + (-a) \left((1 + \epsilon) \frac{2\ln \ell}{\ell} - \frac{1}{2(k-1)}\right) \geq \frac{k-1}{k}(1 - a) R.
\]
Now subtract $-a (k-1) R / k$ to get that the following must hold
\begin{equation}\label{eq1n}
1 - \frac{1}{\ell} + (-a) \underbrace{\left((1 + \epsilon) \frac{2 \ln \ell}{\ell} - \frac{1}{2 (k-1)} - \frac{k - 1}{k} R \right)}_{B} \geq \frac{k-1}{k} R.
\end{equation}
$B$ is negative for large enough $k$, as $\ell \geq k$ and $R \geq 1/2$. Hence
  to minimise the left-hand side of~\eqref{eq1n} we must take $a = - 1 / (k - 1)$ again. Thus it is a sufficient condition on $R$ that
\[
1 - \frac{1}{\ell} + (1 + \epsilon) \frac{2 \ln \ell}{\ell(k-1)} - \frac{1}{2 {(k-1)}^2} - \frac{R}{k}
\geq \frac{k-1}{k} R.
\]
Add $R / k$ and reverse the bound to find the sufficient condition
\[
R \leq 1 - \frac{1}{\ell} +  (1 + \epsilon) \frac{2 \ln \ell}{\ell (k-1)} - \frac{ 1}{2{(k - 1)}^2}.
\]
Rearrange again to find the sufficient condition
\begin{multline*}
R \leq 
1 - \frac{1}{\ell} + \frac{2\ln \ell}{\ell (k-1)}
+ \frac{2 \epsilon(\ell) \ln \ell}{\ell(k-1)} - \frac{1}{2{(k-1)}^2}\\
=
1 - \frac{1}{\ell} + \frac{2\ln\ell}{\ell k} +
\underbrace{\left(
\frac{2 \ln \ell}{\ell k (k-1)}
+ \frac{2 \epsilon(\ell) \ln \ell}{\ell(k-1)} - \frac{1}{2{(k-1)}^2}\right)}_{C}
\end{multline*}

Observe that $C \geq O(\ln \ell / k^2 \ell) - o(\ln \ell / k \ell) - O(1 / k^2) = -o(\ln \ell / \ell k) - O( 1 / k^2)$. (The $o(\ln \ell / \ell k)$ term has a negative constant since it is possible that $\epsilon(\ell)$ is negative.) Hence, since this condition is a sufficient condition on $R$, for large enough $k$,
\[
\alpha_{k\ell} \geq 1 - \frac{1}{\ell} + \frac{2\ln \ell}{k\ell} - o\left( \frac{\ln \ell}{k\ell}\right) - O\left(\frac{1}{k^2} \right). \qedhere
\]
\end{proof}

\subsection{Proof of Proposition~\ref{prop:nontrivial}}

Now, we prove Proposition~\ref{prop:nontrivial}.

\propnontrivial*

This proposition serves a role analogous to~\cite[Corollary 6]{FJ97} (which is equivalent to the case $k = \ell$). We believe that our proof of this fact is simpler; also the direct generalisation of the proof in~\cite{FJ97} does not seem to work for $k$ much smaller than $\ell$.
We will first need a technical lemma.

\begin{lemma}
    For $-1 \leq a \leq 0$, we have $N_\ell(a) \leq 1 / \ell$, with equality only at $a = 0$.
\end{lemma}

\begin{proof}
    First, recall that $N_\ell(0) = 1 / \ell$.
    Note that $N_\ell(a)$ is the probability that $\arg \max_i x_i = \arg \max_i ax_i +  \sqrt{1 - a^2} y_i$, where $x_1, \ldots, x_\ell, y_1, \ldots, y_\ell \sim \mathcal{N}(0, 1)$ are i.i.d.~variables. Now, suppose without loss of generality that the values of ${(x_i)}_{i \in [\ell]}$ are fixed, and in particular $x_1 > \cdots > x_\ell$ (the inequalities are strict with probability 1). Letting $A = - a / \sqrt{1 - a^2} \geq 0$, we have that $N_\ell(a)$ is just $\Pr [\arg \max_i y_i - A x_i = 1]$. Now, fix $y_1$, and note that conditional on this, the probability above becomes
    \[
    \prod_{i = 2}^\ell \Pr[ y_1 - Ax_1 > y_i - Ax_i]
    =
    \prod_{i = 2}^\ell \Pr[ y_i < y_1 - A(x_1 - x_i)]
    \]
    by independence. As $x_1 - x_i > 0$, term-by-term this probability is maximised at $A = a = 0$ (and only there). Since all the probabilities are nonzero, we get that the only $a$ that minimises this expression is $a = 0$. Hence, after integrating over all possible choices of $y_1$, we get that $N_\ell(a) \leq N_\ell(0) = 1 / \ell$, with equality only at $a = 0$.
\end{proof}

\begin{proof}[Proof of Proposition~\ref{prop:nontrivial}]
We wish to prove that, for $-1 /(k-1) \leq a < 1$,
\[
\frac{k - 1}{k} \frac{\ell - 1}{\ell}(1 - a) < 1 - N_\ell(a).
\]
For $a \in \interval[open right]{0}{1}$, this follows immediately by Lemma~\ref{lem:nonengativeBound}, so assume $a < 0$. For such $a$, we know that $N_\ell(a) < 1 / \ell$. Furthermore, note that $0 \leq (k-1)(1-a) / k \leq 1$ for our choice of $a$, so
\[
\frac{k - 1}{k} \frac{\ell - 1}{\ell} (1 - a) \leq \frac{\ell - 1}{\ell} = 1 - \frac{1}{\ell} < 1 - N_\ell(a).\qedhere
\]
\end{proof}

\section{Derandomisation}\label{sec:derandom}

In this section, we will show how to derandomise our algorithm from Theorem~\ref{thm:main} and thus establish Theorem~\ref{thm:maindet}.\footnote{Throughout we will ignore issues of real precision.}
We will use the following result established in~\cite{nz24:arxiv-bipartite}.

\begin{theorem}\label{thm:derandGaussian}
    Fix a constant $d$. There exists an algorithm that does the following. Suppose we are given $n, m \in \mathbb{N}$, $\va_{ij} \in \mathbb{R}^n$ and $b_{ij}, \epsilon \in \mathbb{R}$ for all $i \in [m], j \in [d]$. Suppose $\vx = (x_1, \ldots, x_d) \sim \mathcal{N}(\mathbf{0}, \mathbf{I}_d)$ and that
    \[
    \sum_{i =1 }^m \Pr_{\mathbf{x}}\left[ \bigwedge_{j = 1}^d \va_{ij} \cdot \vx > b_{ij} \right] \geq \alpha
    \]
    for some $\alpha \in \mathbb{R}$. Then the algorithm computes some particular $\vx^* = (x_1^*, \ldots, x_d^*) \in \mathbb{R}^d$ such that
    \[
    \sum_{i =1 }^m \left[ \bigwedge_{j = 1}^d \va_{ij} \cdot \vx^* > b_{ij} \right] \geq \alpha - \epsilon,
    \]
    in polynomial time with respect to $n, m, 1/\epsilon$.
\end{theorem}

\maindet*

\begin{proof}
    Let $G=(V,E)$, where $V = [n]$ and $m = |E|$. Assume that $n \geq k$ (otherwise simply check all possible colourings). By the analysis of our randomised algorithm from Theorem~\ref{thm:main}, using SDP we can find, in polynomial time with respect to $G$ and $\log(1 / \epsilon) + O(1)$, a set of vectors $\va_1, \ldots, \va_n$ such that $\va_i \cdot \va_i = 1$, $\va_i \cdot \va_j \geq -1/(k-1)$ for $i \neq j$ and, if $\vx_1, \ldots, \vx_\ell \sim \mathcal{N}(\mathbf{0}, \mathbf{I}_n)$ are normally distributed variables, then
    \[
    \frac{1}{m} \sum_{(i, j) \in E} \Pr_{\vx_1, \ldots, \vx_\ell} \left[ \arg \max_{c} \vx_c \cdot \va_i \neq 
    \arg \max_c \vx_c \cdot \va_j \right] \geq \alpha_{k\ell} \rho - \frac{\epsilon}{2}.
    \]
    Now, let $\vx = (\vx_1, \ldots, \vx_\ell) \sim \mathcal{N}(\mathbf{0}, \mathbf{I}_{\ell n})$, and define $\va_{ic}$ such that $\vx \cdot \va_{ic} = \vx_c \cdot \va_i$; in other words, pad out $\va_i$ with $(\ell - 1) n$ zeroes. We first claim that the event
    \[
    \arg \max_{c} \vx_c \cdot \va_i \neq \arg \max_c \vx_c \cdot \va_j
    \]
    can be seen as the disjoint union of $\ell(\ell-1)$ intersections of $2(\ell - 1)$ hyperplanes in the space of $\vx$. To express it in this way, first fix the value of the respective sides to $c_0 \neq c_1$, where $c_0, c_1 \in [\ell]$, in $\ell (\ell-1)$ ways. Observe that the event that $\arg \max_c \vx_c \cdot \va_i = c_0$ is the same as
    \[
    \bigwedge_{c \neq c_0} \vx_{c_0} \cdot \va_i > \vx_{c} \cdot \va_i.
    \]
    Now, using the notation from before, this is equivalent to 
    \[
    \bigwedge_{c \neq c_0} \vx \cdot (\va_{i c_0} - \va_{i c}) > 0.
    \]
    It follows that
    \begin{multline*}
    \alpha_{k\ell} \rho - \frac{\epsilon}{2} \leq 
    \frac{1}{m} \sum_{(i, j) \in E} \Pr_{\vx_1, \ldots, \vx_\ell} \left[ \arg \max_{c} \vx_c \cdot \va_i \neq 
    \arg \max_c \vx_c \cdot \va_j \right] \\
    =
    \frac{1}{m} \sum_{(i, j) \in E} \sum_{c_0 \neq c_1} \Pr_{\vx_1, \ldots, \vx_\ell} \left[
    \bigwedge_{c \neq c_0} \vx \cdot (\va_{ic_0} - \va_{ic}) > 0
    \land
    \bigwedge_{c \neq c_1} \vx \cdot (\va_{jc_1} - \va_{jc}) > 0
    \right].
    \end{multline*}
    By Theorem~\ref{thm:derandGaussian} for $d = 2(\ell - 1)$, in polynomial time with respect to $n, m, 1 / \epsilon$, we can find particular values $\vx^*$ such that 
    \begin{multline*}
    \frac{1}{m} \sum_{(i, j) \in E} \sum_{c_0 \neq c_1} \left[
    \bigwedge_{c \neq c_0} \vx^* \cdot (\va_{ic_0} - \va_{ic}) > 0
    \land
    \bigwedge_{c \neq c_1} \vx^* \cdot (\va_{jc_1} - \va_{jc}) > 0
    \right] \\
    \geq \alpha_{k\ell} \rho - \frac{\epsilon}{2} - \frac{\epsilon}{2} = \alpha_{k\ell} \rho - \epsilon.
    \end{multline*}
    Defining $(\vx_1^*, \ldots, \vx_\ell^*) = \vx^*$, this is equivalent to
    \[
    \frac{1}{m} \sum_{(i, j) \in E} [\arg \max_c \vx_c^* \cdot \va_i \neq
    \arg \max_c \vx_c^* \cdot \va_j] \geq \alpha_{k\ell} \rho - \epsilon.
    \]
    In other words, if we set the colour of vertex $i$ to $\arg \max_c \vx_c^* \cdot \va_i$, then the resulting $\ell$-colouring will have value $\alpha_{k\ell} \rho - \epsilon$, as required.
\end{proof}

\section{\texorpdfstring{Algorithm for fixed $\bm{k}$ and large $\bm{\ell}$}{Algorithm for fixed k and large l}}

We show the following theorem.

\thmbigk*

Note that Theorem~\ref{thm:bigk} together with Theorem~\ref{thm:hardness} do not contradict the \UGC~and $\P \neq \NP$, since Theorem~\ref{thm:hardness} only works for bounded $\ell$.
\begin{proof}
As always, we will only care about when $\ell$ grows large. We will first give a randomised algorithm, and then derandomise it. We solve the same semi-definite program as in Theorem~\ref{thm:main} (which is also the same as in~\cite{FJ97, Karger98:jacm}), i.e.
\begin{maxi*}|s|
{}{\frac{1}{m} \sum_{(i, j) \in E} \frac{k - 1}{k}(1 - \va_i \cdot \va_j)} 
{}{}
\addConstraint{\va_i \cdot \va_i = 1}{}
\addConstraint{\va_i \cdot \va_j \geq -\frac{1}{k-1}, i \neq j} 
\addConstraint{\va_i \in \mathbb{R}^n}. 
\end{maxi*}
As in Theorem~\ref{thm:main}, the value is at least $\rho$. We now randomly round in the following way: Sample $t = \lfloor \log_2(\ell) \rfloor$ random hyperplanes that pass through the origin $H_1, \ldots, H_t$ in $n$ dimensions; then, to colour node $i$ check on which side of $H_1, \ldots, H_t$ the vector $\va_i$ is, and depending on this assign a unique colour. Note that we use at most $2^t \leq \ell$ colours in this way. Fix an edge $(i, j)$ and consider $a = \va_i \cdot \va_j$; what is the probability that the colours assigned to $i$ and $j$ are different? Note that the probability that $\va_i, \va_j$ are separated by one hyperplane among $H_1, \ldots, H_t$ is just $\frac{1}{\pi} \arccos a$ (this observation is originally from~\cite{GW95}). So the probability that $\va_i, \va_j$ will be separated by at least one hyperplane is
\[
1 - \left(1 - \frac{1}{\pi} \arccos a\right)^t.
\]
Now, the approximation ratio is given by
\[
\alpha_{k\ell}' = \min_{-1 / (k-1) \leq a < 1} \frac{k (1 - (1 - \arccos a / \pi)^t)}{(k-1) (1 - a)}.
\]
We first deal with $a$ around a neighbourhood of 1, similarly to Lemma~\ref{lem:nonengativeBound}. We claim that there exists some $0 < a_k < 1$ such that
\[
\frac{\arccos a}{\pi} \geq \frac{(k-1)(1-a)}{k}
\]
for all $a \in (a_k, 1]$.
Indeed, consider $\arccos a / \pi - (k - 1)(1 - a) / k$. The derivative tends to $-\infty$ as $a \to 1$ (from below), so for some neighbourhood of $1$ the derivative is negative. Suppose $(a_k, 1]$ is this neighbourhood. Thus the function is decreasing on this interval. Since the function is equal to 0 at $1$, our conclusion follows.

Now, observe that $\arccos a / \pi \in [0, 1]$, hence
\[
1 - \left( 1 - \frac{\arccos a}{\pi}\right)^t \geq 1 - \left(1 - \frac{\arccos a}{\pi}\right) = \frac{\arccos a}{\pi} \geq \frac{(k - 1)(1 - a)}{k}
\]
when $a \in (a_k, 1]$. Hence the expression minimised in the definition of $\alpha_{k\ell}$ is at least 1 whenever $a \in (a_k, 1]$, and thus does not affect the value of $\alpha_{k\ell}'$. We now focus on the case $-1 / (k - 1) < a \leq a_k < 1$.

Define
\[
f(a) = \frac{1 - (1 - \arccos a / \pi)^t}{1 - a}.
\]
Observe that
\[
(a - 1)^2 f'(a) = 1 - \frac{t}{\pi} \sqrt{\frac{1 - a}{1+a}}\left( 1 - \frac{\arccos(a)}{\pi}\right)^{t-1} - \left(1 - \frac{\arccos(a)}{\pi}\right)^t.
\]
Note that for large enough $t$ (i.e.~large enough $\ell$), we have that $f'(a) > 0$ for $-1/(k-1) \leq a \leq a_k < 1$. (The size required of $\ell$ depends on $a_k$ and hence on $k$.) Hence, 
\[
\alpha_{k\ell}' = \frac{k}{k-1} \frac{1 - \left(1 - \frac{1}{\pi} \arccos \left(-\frac{1}{k-1}\right)\right)^t}{\left(1 + \frac{1}{k-1}\right)} = 1 - \left(1 - \frac{1}{\pi} \arccos\left(- \frac{1}{k-1}\right)\right)^t.
\]
(This value is indeed less than 1, so $a\in (a_k, 1]$ did not matter.) We now observe that for any fixed $k > 2$,
\[
X_k \coloneqq 1 - \frac{1}{\pi}\arccos\left(- \frac{1}{k-1}\right) \in \left(0, \frac{1}{2}\right).
\]
Define $u_k = -\log_2(X_k) > 1$. Hence
\[
\alpha_{k\ell}' = 1 - X_k^t = 1 - 2^{\log_2X_k \lfloor \log_2 \ell \rfloor}.
\]
Observe that
\begin{multline*}
\log_2 X_k \lfloor \log_2 \ell \rfloor \leq \log_2X_k (-1 +  \log_2 \ell) = -\log_2 X_k + \log_2 X_k \log_2 \ell \\
= -\log_2 X_k - u_k \log_2 \ell.
\end{multline*}
Thus
\[
\alpha_{k\ell}' \geq 1 - 1 / (X_k \ell^{u_k}) = 1 - O(1 / \ell^{u_k}).
\]

We now turn to derandomising this algorithm. It is sufficient to show that the event that two vectors $\va_i, \va_j$ are properly cut by one of the $H_1, \ldots, H_t$ hyperplanes is the disjoint union of the intersection of constantly many half-spaces in some multivariate normal probability distribution. Then, the derandomisation works precisely as for Theorem~\ref{thm:main}, using Theorem~\ref{thm:derandGaussian}. First, we must express our hyperplanes $H_1, \ldots, H_t$ in terms of normal variables. As was first observed by~\cite{FJ97}, a uniformly random hyperplane $H_i$ can be sampled by taking the set of points at equal distance between two vectors $\vx_i, \vy_i \sim \mathcal{N}(\mathbf{0}, \mathbf{I}_n)$, and the points to one side or the other of the hyperplane are those points closer (in terms of inner product) to $\vx_i$ or $\vy_i$ respectively. We label the two sides of $H_i$ with $0$ and $1$, with side $0$ containing $\vx_i$ and side 1 containing $\vy_i$, and take the convention that if a vector is on $H_i$ then it is on side $0$. Then the event that $\va_i$ is on side 0 of $H_j$ is
\[
\va_i \cdot \vx_j \geq \va_i \cdot \vy_j.
\]
Call this event $E(i, j, 0)$, and the complementary event $E(i, j, 1)$. If as before we write
\[
\vx = (\vx_1, \ldots, \vx_n, \vy_1, \ldots, \vy_n) \sim \mathcal{N}(\mathbf{0}, \mathbf{I}_{2tn}),
\]
then each event $E(i, j, x)$ is equivalent to $\vx$ belonging to a half-space. Now, the event that vertex $i$ is assigned colour $c$ (call it $E(i, c)$) is equivalent to a conjunction of $t$ of these events, one for each hyperplane $H_1, \ldots, H_t$. Furthermore, the event that $\va_i, \va_j$ are properly cut is equivalent to the disjoint union of at most $2^t\times 2^t$ of conjunctions of these events, namely
\[
\bigvee_{c \neq c'} E(i, c) \land E(j, c').
\]
Hence we can derandomise as before, and our conclusion follows.
\end{proof}

\section{Hardness}\label{sec:hardness}

In this section we use the approach of Khot, Kindler, Mossel and O'Donnell~\cite{KKMO07} and of Guruswami and Sinop~\cite{Guruswami13:toc} respectively to prove the following hardness result.

\hardness*

The conditional bound in the first bullet point matches the bound in Theorem~\ref{thm:main} (up to the
asymptotic error terms), \emph{but only when $\ell$ is bounded by some $M(k)$
that is strictly smaller, asymptotically, than the superpolynomial function
$e^{\sqrt[3]{k}}$}. Throughout this entire section, fix the function $M(k)$ and
$2 \leq k \leq \ell$. 

In the rest of this introduction to Section~\ref{sec:hardness}, we give a brief
overview of the proof of Theorem~\ref{thm:hardness} and how it differs from
existing work. All definitions not stated here explicitly can be found, together
with all details, in later subsections.
We first recall the definition of label cover.

\newcommand{\Va}{\ensuremath{V_A}}
\newcommand{\Vb}{\ensuremath{V_B}}
\newcommand{\Vbp}{\ensuremath{V_{B'}}}

\begin{definition}
An instance of \emph{label cover with $p$-to-1 constraints with domain size $r$}
  is a tuple $I = (V = \Va \cup \Vb, E, \pi)$, where $(\Va \cup \Vb, E)$ is
  a  bipartite graph, and for each edge $(a, b) \in E$ we have a 
  constraint\footnote{We note that alternatively we could have defined our constraints as $p$-to-1 relations $\pi \subseteq [pr] \times [r]$ i.e.~relations where for every $x \in [r]$ there exist exactly $p$ values $y \in [pr]$ such that $(y, x) \in \pi$, and furthermore every $y \in [pr]$ corresponds to exactly one $x \in [r]$ such that $(y, x) \in \pi$. To translate from this view to ours, map every $x \in [pr]$ to that $y \in [r]$ such that $(x, y) \in \pi$; to translate from our view to this one, take the graph of the function we use as a constraint.}
  $\pi_{a,b} : [pr] \to [r]$ that is $p$-to-1; i.e., for
  every $x\in [r]$ there are precisely $p$ values $y\in [pr]$ such that
  $x = \pi_{a, b}(y)$.
  We call the instance left-regular if every vertex $a \in \Va$ has the same degree. A solution to this instance is a mapping $c$ that
  takes $\Va$ to $[r]$ and $\Vb$ to $[pr]$. The value of the solution is the
  proportion of edges $(a, b) \in E$ with 
  $c(a) = \pi_{a,b}(c(b))$. 
  The value of the instance is the maximum value of any solution.

The problem of $(1 - \eta, \eta)$-approximating a label cover with $p$-to-1 constraints with domain size $r$ is the following: Given an instance $I$ of label cover with $p$-to-1 constraints and domain size $r$, decide if its value is at least $1 - \eta$ or at most $\eta$.
\end{definition}

The following two are not the original forms of the Unique Games Conjecture or the 2-to-1 Theorem, but they are equivalent to them due to the reductions in~\cite{KR08} --- the original forms considered weighted non-left-regular label cover instances.

\begin{conjecture}[Unique Games Conjecture (\UGC)~\cite{Khot02stoc}]
For every small $\eta > 0$, there exists an $r \in \mathbb{N}$ such that it is \NP-hard to $(1 - \eta, \eta)$-approximate a left-regular label cover with 1-to-1 constraints and domain size $r$.
\end{conjecture}

\begin{theorem}[2-to-1 Theorem~\cite{Khot23:annals}]
For every small $\eta > 0$, there exists an $r \in \mathbb{N}$ such that it is \NP-hard to $(1 - \eta, \eta)$-approximate a left-regular label cover with 2-to-1 constraints and domain size $r$.
\end{theorem}

Theorem~\ref{thm:hardness} follows 
from Propositions~\ref{prop:bound} and~\ref{prop:reduction} stated below.
Proposition~\ref{prop:bound} serves the same role as~\cite[Proposition
12]{KKMO07} and~\cite[Proposition 3.21]{Guruswami13:toc}; we reprove it here
since the promise version does not immediately follow from the non-promise
version, and since Lemma~\ref{lem:finalPart} (needed in the proof of
Proposition~\ref{prop:bound}) fixes a small bug in the published
proofs. On the other hand, Proposition~\ref{prop:reduction} is a standard PCP
construction, analogous to~\cite[Section 11.4]{KKMO07} and~\cite[Section
3.4]{Guruswami13:toc}; we prove it here in a unified way (covering
simultaneously 2-to-1 and 1-to-1 constraints) rather than repeating most of the
proof twice. The way the unification works is similar to the proof in~\cite{GS20:icalp}.

\begin{restatable}{proposition}{propbound}\label{prop:bound}
    Fix $p \in \{1, 2\}$ and $2 \leq k \leq \ell$ such that $\ell \leq M(k)$.
    Suppose $T$ is a symmetric Markov operator on $[k^p]$ with spectral radius
    $0 < c / (k - 1) < 1$, where $c \leq 4$. Then there exists $\tau > 0$ and $d \in \mathbb{N}$ such that for any $f : {[k^p]}^{r} \to \Delta_\ell$ with $\Inf_i^{\leq d}(f) \leq \tau$ for all $i \in [r]$, we have
    \[
    \langle f, T^{\otimes n} f \rangle \geq \frac{1}{\ell} - \frac{2 c \ln \ell}{k \ell} - 
    o\left(\frac{\ln \ell}{k \ell}\right).
    \]
\end{restatable}

\begin{restatable}{proposition}{propreduction}\label{prop:reduction}
    Let $p \in \{1, 2\}$ and $2 \leq k \leq \ell$. Assume that there is a
    colourful symmetric Markov operator $T$ on $[k^p]$ and $\tau > 0, d \in
    \mathbb{N}$ such that for any $f : {[k^p]}^r \to \Delta_\ell$ with
    $\Inf_i^{\leq d}(f) \leq \tau$ for all $i \in [n]$, we have that $\langle f,
    T^{\otimes r} f \rangle \geq 1 - \beta$ for some $\beta \in (0, 1)$. Assume
    further that all the values of $T$ are nonnegative integer multiples of a
    rational number $L$. Then, assuming that $(1 - \eta, \eta)$-approximating a
    label cover instance with $p$-to-1 constraints is \NP-hard, for any small
    enough $\epsilon$, it is \NP-hard to decide whether a given graph $G$ has a $k$-colouring of value $1 - \epsilon$, or not even an $\ell$-colouring of value $\beta + \epsilon$. 
\end{restatable}

The Markov operators mentioned above will be given in the following theorems.
The notion of a ``colourful'' Markov chain merely unifies two properties that we
are interested in for our PCP construction. This property appears without a name
also in~\cite[Lemma 10]{GS20:icalp}.\footnote{\cite[Lemma 3.8]{Guruswami13:toc}
constructs a Markov chain that is colourful in our terminology, and implicitly
uses this fact in their PCP construction. However their lemma states merely that the Markov chain has diagonal elements equal to zero, which is not by itself enough to make the PCP reduction work.}

\begin{restatable}{theorem}{thmnoiseone}\label{thm:noise1}
    Fix $k \geq 2$. The Bonami-Beckner operator $T_{-1 / (k - 1)}$ on $[k]$ is a
    symmetric Markov operator that is colourful with spectral radius $1 / (k -
    1)$. Furthermore, all the values in the matrix are nonnegative integer multiples of $1 / (k - 1)$.
\end{restatable}

\begin{restatable}[{\cite[Lemma 3.8]{Guruswami13:toc}}]{theorem}{thmnoisetwo}\label{thm:noise2}
    Fix $k \geq 6$. There exists a colourful symmetric Markov operator on $[k^2]$ whose spectral radius is at most $4 / (k - 1)$. Furthermore, all the values in the matrix are nonnegative integer multiples of $1 / (k-1)(k-2)(k-3)$.
\end{restatable}

The proof of Theorem~\ref{thm:hardness} thus follows by combining all the facts listed above, together with the observation that we can take $k$ to be large, since the fact that $\ell$ is bounded by a function of $k$ means that the asymptotic term can be increased for small $k$ to make the theorem hold for small $k$.

\subsection{Fourier-analytic notions}

We closely follow the exposition of Fourier analysis on discrete domains from~\cite{Dinur06:STOC}. We also include some results from~\cite{Guruswami13:toc}. We diverge from these only in that they number their colours $0, \ldots, k - 1$, whereas we number them $1, \ldots, k$; also we simplify the notation for our Fourier coefficients.

We will be looking in general at functions of the form $[D]^r \to [\ell]$ or
$[D]^r \to [0, 1]$. Note that $[\ell]$ can be naturally embedded in the set of
probability distributions over $[\ell]$, which can be seen as the set
$\Delta_\ell = \{ (x_1, \ldots, x_\ell) \mid \sum_i x_i = 1, x_i \geq 0 \}$.
Such functions form a vector space under point-wise addition and multiplication, and they have a natural inner product, namely
\[
\langle f, g \rangle = \Exp_{\vx \sim \mathcal{U}([D]^r)} \left[ f(\vx) \cdot g(\vx) \right].
\]
The inner product induces a norm $||\cdot||_2$. In general we will assume that a variable $\vx$ that is mentioned in an expectation will be taken uniformly at random from an appropriate set. Observe that if $f(\vx) = (f^1(\vx), \ldots, f^\ell(\vx))$ and $g(\vx) = (g^1(\vx), \ldots, g^\ell(\vx))$, then
\[
\langle f, g \rangle = \sum_{i = 1}^\ell \langle f^i, g^i \rangle.
\]
So in particular $|| f ||_2^2 = \sum_{i = 1}^\ell || f_i ||_2^2$.

For any two functions $f : A \to \R, g : B \to \R$ we define $f \otimes g : A
\times B \to R$ by $(f \otimes g)(a, b) = f(a) g(b)$; thus $\otimes$ is a tensor product.
For every $[D]$, fix some orthonormal (under $\langle \cdot, \cdot \rangle$\footnote{The only difference between orthonormality of function $[D] \to \R$ under $\langle \cdot,\cdot\rangle$ and of $\R^D$ under the normal inner product is a matter of normalisation.}) basis of the set of functions $[D] \to \R$, namely some functions $\valpha_1, \ldots, \valpha_D$, such that $\valpha_1(1) = \ldots = \valpha_1(D) = 1$.\footnote{In Lemma~\ref{lem:Sanglerelation}, and only there, we will need a particular choice of $\alpha_1, \ldots, \alpha_D$. However, the lemma does not mention any quantity dependant on this choice in its statement, so the choice is ``contained'' within that lemma.}
For $\vx \in [D]^r$ we define $\valpha_\vx : [D]^r \to \R$ by
\[
\valpha_\vx = \valpha_{x_1} \otimes \cdots \otimes \valpha_{x_r}
\]
or equivalently
\[
\valpha_\vx(y_1, \ldots, y_r) = \prod_{i = 1}^r \valpha_{x_i}(y_i).
\]
It can be seen that for $\vx \neq \vy$ we have that $\valpha_\vx \perp
\valpha_\vy$. Furthermore, since there are $[D]^r$ such function (i.e.~the same as the dimension of the set of functions from $[D]^r$ to $\R$), for any $f : [D]^r \to \R$ we have
\[
f = \sum_{\vx \in [D]^r} \langle f, \valpha_\vx \rangle \valpha_\vx = \sum_{\vx \in [D]^r} \hat{f}(\vx) \valpha_\vx,
\]
where $\hat{f}(\vx) = \langle f, \valpha_\vx \rangle$. This $\hat{f}(\vx)$ is known as a Fourier coefficient.
Clearly $\widehat{af + bg} = a \hat{f} + b \hat{g}$, hence $\hat{\cdot}$ is linear. 
We have a version of Parseval's identity now, since $\{ \valpha_\vx\}_\vx$ forms a basis:
\[
|| f ||_2^2 = \sum_{ \vx \in [D]^r } \hat{f}^2(\vx).
\]
Indeed, in general we have that
\[
\langle f, g \rangle = \sum_{\vx \in [D]^r} \hat{f}(\vx) \hat{g}(\vx).
\]
We also generalise our notion of Fourier coefficient to functions $f : [D]^r \to \R^\ell$. If $f(\vx) = (f_1(\vx), \ldots, f_\ell(\vx))$, then
\[
\hat{f}(\vx) = (\hat{f_1}(\vx), \ldots, \hat{f_\ell}(\vx)).
\]
With these, Parseval's inequality generalises to
\[
|| f ||_2^2 = \sum_{i = 1}^\ell || f_i ||_2^2 = \sum_{i = 1}^\ell \sum_{\vx \in [D]^r} \hat{f_i}^2(\vx)
= \sum_{\vx \in [D]^r} | \hat{f}(\vx) |^2.
\]

We now introduce the notion of \emph{low-degree influence}. First for $\vx \in [D]^r$, we let $|\vx|$ to be the number of coordinates $i \in [r]$ such that $\vx_i \neq 1$. With this in hand, for $f : [D]^r \to \R$, $i \in [r]$, $d \leq r$, we define
\[
\Inf_{i}^{\leq d}(f) = \sum_{\substack{\vx \in [D]^r \\ \vx_i \neq 1 \\ |\vx| \leq d}} \hat{f}^2(\vx).
\]
This definition does not depend on the basis $\valpha_1, \ldots, \valpha_D$ taken initially (only that $\valpha_1(x) = 1$); for details see~\cite[Definition 2.5]{Dinur06:STOC}.

We now generalise this definition to functions of the form $f : [D]^r \to \R^\ell$; for such a function suppose $f(\vx) = ( f_1(\vx), \ldots, f_d(\vx))$. Then, for $i \in [r]$, $d \leq r$, we define
\[
\Inf_i^{\leq d}(f) = \sum_{j = 1}^\ell \Inf_i^{\leq d}(f_j).
\]
Observe that
\[
\Inf_i^{\leq d}(f) =
\sum_{j = 1}^\ell
\sum_{\substack{\vx \in [D]^r \\ \vx_i \neq 1 \\ |\vx| \leq d}} \hat{f_j}^2(\vx)
=
\sum_{\substack{\vx \in [D]^r \\ \vx_i \neq 1 \\ |\vx| \leq d}} 
\sum_{j = 1}^\ell
\hat{f_j}^2(\vx)
=
\sum_{\substack{\vx \in [D]^r \\ \vx_i \neq 1 \\ |\vx| \leq d}} 
| \hat{f}(\vx) |^2.
\]

Next, we deduce some classic inequalities for sums of low-level influences. Consider $f : [D]^r \to \Delta_\ell$. Note that
\[
|| f||_2^2 = \Exp_\vx [ |f(\vx)|^2 ] \leq 1,
\]
since for any $\vy \in \Delta_\ell$ we have $|\vy|^2 \leq 1$. Thus, by Parseval's identity, we have
\[
1 \geq || f ||_2^2 = \sum_{\vx \in [D]} | \hat{f}(\vx) |^2.
\]
Observe that the formula giving $\Inf_i^{\leq d}(f)$ as a sum of square lengths of Fourier coefficients contains each term in the sum above at most $d$ times; since the sum is at most 1, we derive that $\Inf_i^{\leq d}(f) \leq d$.

\paragraph*{Minor operations}
Consider any vector $\vx \in [D]^r$ and a function $\pi : [s] \to [r]$. Then we define $\vx^\pi \in [D]^s$ by
\[
\vx^\pi = (\vx_{\pi(1)}, \ldots, \vx_{\pi(s)}).
\]
Furthermore, consider any function $f : [D]^r \to \R^d$, and let $\pi :[r] \to [s]$. Then define $f^\pi : [D]^s \to \R^d$ by
\[
f^\pi(\vx) = f(\vx^\pi) = f(x_{\pi(1)}, \ldots, x_{\pi(r)}).
\]
Observe that when the function $\pi : [r] \to [r]$ is a bijection, and for function $f, g : [D]^{r} \to \R^d$, we have that
\[
\langle f^\pi, g^\pi \rangle=
\Exp_\vx [ \langle f^\pi(\vx), g^\pi(\vx) \rangle ] =
\Exp_\vx [ \langle f(\vx^\pi), g(\vx^\pi) \rangle ] =
\Exp_\vx [ \langle f(\vx), g(\vx) \rangle ] =
\langle f, g \rangle.
\]
Next, observe that for such $\pi$,
\[
\valpha_\vx^{\pi}(y_1, \ldots, y_r) = \valpha_\vx(y_{\pi(1)}, \ldots, y_{\pi(r)})
=
\prod_{i = 1}^r \alpha_{x_i}(y_{\pi(i)})
=
\prod_{i = 1}^r \alpha_{x_{\pi^{-1}(i)}} (y_i)
= \valpha_{\vx^{\pi^{-1}}}(y_1, \ldots, y_r).
\]
Hence $\valpha_\vx^\pi = \valpha_{\vx^{\pi^{-1}}}$.
Furthermore this implies that for bijective $\pi$,
\[
\widehat{f^\pi}(\vx)
=
\langle f^\pi, \alpha_\vx \rangle
=
\langle f, \alpha_{\vx}^{\pi^{-1}} \rangle
=
\langle f, \alpha_{\vx^\pi} \rangle
= \hat{f}(\vx^\pi).
\]

\paragraph*{Coordinate regrouping lemmas}
We reuse and slightly modify the notation from~\cite[Definition
2.6]{Dinur06:STOC}, which reappears in~\cite[Definition 3.22]{Guruswami13:toc}.
The notation there does not include the $(2)$ we use --- this will be important for us since what we do is in greater generality.

We will implicitly use the fact that $[k^2] \cong [k]^2$; fix some arbitrary bijection between the two. For $\vx \in [k]^{2r}$, we define $\overline{\vx}^{(2)} \in [k^2]^r \cong ([k]^2)^r$ by
\[
\overline{\vx}^{(2)} = ((x_1, x_2), \ldots, (x_{2r-1}, x_{2r})).
\]
Conversely, for $\vx \in [k^2]^r \cong ([k]^2)^r$, where $\vx = ((x_1, x_1'), \ldots, (x_n, x_n'))$, we have
\[
\underline{\vx}_{(2)} = (x_1, x_1', \ldots, x_n, x_n').
\]
From the definitions, $\overline{\underline{\vx}}^{(2)}_{(2)} = \vx$. Now, we define these transformations on the inputs for functions, as follows. For $f : [k]^{2r} \to \R^d, g : [k^2]^{r} \to \R^d$, we have
\begin{align*}
    \overline{f}^{(2)}(\vx) & = f(\underline{\vx}_{(2)}), \\
    \underline{g}_{(2)}(\vx) & = g(\overline{\vx}^{(2)}).
\end{align*}
Again, $\overline{\underline{f}}^{(2)}_{(2)} = f$. Next, we recall the following result.

\begin{lemma}[{\cite[Claim 2.7]{Dinur06:STOC}}]
    For $f : [k]^{2r} \to \R$, $i \in [r]$, and $1 \leq d \leq r$, we have
    \[
    \Inf_i^{\leq d}(\overline{f}^{(2)}) \leq \Inf_{2i - 1}^{\leq d}(f) + \Inf_{2i}^{\leq d}(f).
    \]
\end{lemma}

The following is an obvious corollary.

\begin{corollary}
    For $f : [k]^{2r} \to \R^\ell$, $i \in [r]$, and $1 \leq d \leq r$, we have
    \[
    \Inf_i^{\leq d}(\overline{f}^{(2)}) \leq \Inf_{2i - 1}^{\leq d}(f) + \Inf_{2i}^{\leq d}(f).
    \]
\end{corollary}
\begin{proof}
    Apply the previous lemma to each term in the sum that defines $\Inf_i^{\leq d}$ for functions to $\R^\ell$.
\end{proof}

Now, we define $\overline{\vx}^{(1)} = \underline{\vx}_{(1)} = \vx$ and likewise $\overline{f}^{(1)} = \underline{f}_{(1)} = f$, and thus immediately
\[
\Inf_i^{\leq d}(\overline{f}^{(1)}) \leq \Inf_i^{\leq d}(f).
\]
Therefore we find that
\begin{corollary}\label{corr:relabelling}
    Fix $p \in \{1, 2\}$. For $f : [k]^{pr} \to \R^\ell$, $i \in [r]$ and $d \in [r]$, we have
    \[
    \Inf_i^{\leq d}(\overline{f}^{(p)}) \leq \sum_{j = p(i - 1) + 1}^{pi} \Inf_{j}^{\leq d}(f).
    \]
\end{corollary}

This setup will allow us to painlessly unify the proofs of~\cite{KKMO07}
and~\cite{Guruswami13:toc}. We note that Corollary~\ref{corr:relabelling} is a special case (i.e.~for $p \in \{1, 2 \}$) of~\cite[Lemma 6]{GS20:icalp}.

\subsection{Markov operators}

We recount some definitions for Markov operators, following~\cite{Dinur06:STOC,Guruswami13:toc}, as well as an important result from~\cite{Guruswami13:toc}.

A Markov operator $T$ on $[D]$ is a $D \times D$ stochastic matrix; we say that the operator is symmetric if $T$ is symmetric. For such an operator, we say that the spectral radius $r(T)$ is the second-largest absolute value of any eigenvalue (such an operator has an eigenvalue equal to 1, and all its eigenvalues must be at most 1 in absolute value). We will usually let $T(x \leftrightarrow y)$ denote the probability that $x$ goes to $y$ i.e.~the element at position $(x, y)$ in $T$.

Such an operator operates on the space of functions $f : [D] \to \R^d$, in the following way. For $x \in [D]$, we let $T(x)$ be the distribution associated with row (or equivalently column) $x$ in $T$. Then we have
\[
(Tf)(x) = \Exp_{y \sim T(x)} [f(y)].
\]
Note that $T$ acts linearly on $f$, since $T(af + bg) = a (Tf) + b(Tg)$ by linearity of expectation. Observe that if $f(x) = (f^1(x), \ldots, f^d(x))$, then
\[
(Tf)(x) = ((Tf^1)(x), \ldots, (Tf^d)(x)).
\]

For any two operators $T, T'$ on $[D], [D']$ respectively, we define the operator $T \otimes T'$ on $[D] \times [D']$ as being the matrix which, at position $((x, x'), (y, y'))$ for $x, y \in [D], x', y' \in [D']$ has value $T(x \leftrightarrow y) T(x' \leftrightarrow y')$. (In other words, $T \otimes T'$ is the Kronecker product of $T$ and $T'$.)

Furthermore, define $T^{\otimes r} = T \otimes \ldots \otimes T$ where $T$ is multiplied $r$ times. Observe that $T^{\otimes r}$ acts on a function $f : [D]^r \to \R^d$ in the following way. Let $T^{\otimes r}(\vx)$ be the product distribution over $[D]^r$ which gives $\vy$ probability $\prod_{i = 1}^r T(x_i \leftrightarrow y_i)$. Then
\[
(Tf)(\vx) = \Exp_{\vy \sim T^{\otimes r}(\vx)}[ f(\vy) ].
\]
Furthermore, we see immediately that $(T \otimes T')(f \otimes f') = (T f)
\otimes (T g)$. Thus, in particular,
\[
T^{\otimes r} \valpha_\vx = \bigotimes_{i = 1}^r T \valpha_{x_i}.
\]

We observe that for a symmetric Markov operator $T$, we have $\langle T f, g \rangle = \langle f, T g \rangle$.

The following unifies concepts from~\cite{KKMO07,Guruswami13:toc}. The idea also appears without a name in~\cite{GS20:icalp}.

\begin{definition}
    Consider a symmetric Markov chain $T$ on $[k^p] \cong [k]^p$. We say that the Markov chain is colourful if, for any $x_1, \ldots, x_p, y_1,\ldots  y_p \in [k]$ such that $T((x_1, \ldots, x_p) \leftrightarrow (y_1, \ldots, y_p)) > 0$, we have $\{x_1, \ldots, x_p\} \cap \{y_1, \ldots, y_p\} = \emptyset$.
\end{definition}

We now introduce two very important colourful operators. The first is (a special case of) the \emph{Bonami-Beckner operator}.

\begin{definition}
    Fix $D \geq 2$. For each $- 1 / (D - 1) \leq \rho \leq 1$ we define the Bonami-Beckner operator by
    \[
    T_{\rho} = \rho \mathbf{I}_D + \frac{1 - \rho}{D} \mathbf{J}_D.
    \]
    The matrix $\mathbf{J}_D$ is the all-ones matrix of size $D \times D$. This operator is clearly symmetric and doubly stochastic i.e.~it is a
    symmetric Markov operator. Its eigenvalues are $1$ and $(k - 1)$ copies of
    $\rho$. Furthermore, any vector $\vx$ whose sum is zero (i.e.~is
    perpendicular to the all-ones vector) is an eigenvector of $T_\rho$ with eigenvalue $\rho$, as
    \[
    T_\rho \vx = \rho \vx + \frac{1 - \rho}{D} \mathbf{J}_D \vx = \rho \vx.
    \]
\end{definition}

\thmnoiseone*

\begin{proof}
    The operator is given by the symmetric matrix whose diagonal elements are zero (this is sufficient for colourfulness), and whose off-diagonal elements are $1 / (k - 1)$. The eigenvalues of this operator are $1$ and $(k - 1)$ copies of $- 1 / (k - 1)$, hence the spectral radius is $|- 1 / (k - 1) | = 1 / (k - 1)$.
\end{proof}

\thmnoisetwo*

(We note that~\cite[Lemma 3.8]{Guruswami13:toc} does not explicitly state the fact that the operator is colourful, or that its elements are multiples of $ 1/ (k-1)(k-2)(k-3)$, but these are easy to observe.)

Given these definitions, we also define the notion of noise stability.

\begin{definition}
    Let $f : [D]^r \to \R$. Then we define, for $-1 / (D - 1) \leq \rho \leq 1$, the noise stability of $f$ as
    \[
    \S_{\rho}(f) = \langle f, T_\rho^{\otimes r} f \rangle.
    \]
    Equivalently,
    \[
    \S_{\rho}(f) = \sum_{\vx \in [D]^r} \rho^{|\vx|} \hat{f}^2(\vx),
    \]
    regardless of the choice of $\alpha_1, \ldots, \alpha_D$ provided $\alpha_1(x) = 1$.
\end{definition}

\subsection{MOO theorem, bounds}

We introduce some simple definitions from~\cite{KKMO07}. (In fact, the exact definition of the quantity $\Lambda_\rho(\mu)$ is not needed, only the bound in Theorem~\ref{thm:estimate}.)

\begin{definition}
    Fix $-1 \leq \rho \leq 1$ and $0 \leq \mu \leq 1$. Let $u$ be some value such that if $x \sim \mathcal{N}(0, 1)$, then $\Pr[x \leq u] = \mu$. Let $\mathbf{\Sigma} = \begin{pmatrix} 1 & \rho \\ \rho & 1 \end{pmatrix}$, and suppose $(x, y) \sim \mathcal{N}(\mathbf{0}, \mathbf{\Sigma})$. Then, we define
    \[
    \Lambda_\rho(\mu) = \Pr[ x \leq u, y \leq u ].
    \]
\end{definition}

We will need the following theorem to prove the main technical lemma. It follows
from what is called the MOO Theorem (i.e.~Mossell, O'Donnell, Oleszkiewicz)
in~\cite{KKMO07}, proved originally by Mossel, O'Donnel and
Oleszkiewicz~\cite{Mossel10:ann}, together with~\cite[Proposition 13]{KKMO07}.

\begin{theorem}\label{thm:modifiedMOO}
    For any $k \geq 2$, $0 \leq \rho < 1$ and $\epsilon > 0$, there exists $\tau
    > 0, d \in \mathbb{N}$ such that the following holds. Suppose $f : {[k]}^r
    \to [0, 1]$ is such that $\Inf_i^{\leq d}(f) \leq \tau$ for all $i \in [n]$, and let $\mu = \Exp[f]$. Then
    \[
    \S_\rho(f) \leq \Lambda_\rho(\mu) + \epsilon.
    \]
\end{theorem}

In order to link this theorem with our setting, we will need the following relation between $\langle f, T^{\otimes r} f \rangle$ for some symmetric Markov operator $T$ with spectral radius $\rho$ and $\S_\rho(f)$. The following bound generalises the first step in proving~\cite[Proposition 12]{KKMO07}, following the first step in the proof of~\cite[Propositio 3.21]{Guruswami13:toc}.

\begin{lemma}\label{lem:Sanglerelation}
Let $T$ be a symmetric Markov operator on $[D]$ with spectral radius $0 < \rho < 1$. For any $f : [D]^r \to [0, 1]$, where $\mu = \Exp[f]$, we have that
\[
\langle f, T^{\otimes r} f \rangle \geq \mu^2 - \S_\rho(f).
\]
\end{lemma}
\begin{proof}
Suppose that $\valpha_1, \ldots, \valpha_D : [D] \to \R$ are orthonormal (with respect to $\langle \cdot, \cdot \rangle$) eigenvectors of $T$, seen as functions, whose eigenvalues are $\lambda_1, \ldots, \lambda_D$. Suppose $\valpha_1 = \mathbf{1}$ is the constant one function, with eigenvalue 1. Then we find that $\valpha_2, \ldots, \valpha_D$ are all perpendicular to the constant-ones function, and have eigenvalue at most $\rho$ in absolute value. Furthermore, we find that $\valpha_1, \ldots, \valpha_D$ are also eigenvectors of $T_\rho$, with eigenvalues $1, \rho, \ldots, \rho$.

Recall that $T^{\otimes r} \valpha_\vx = \bigotimes_{i = 1}^r T \valpha_{x_i}$; now $T \valpha_{x_i} = \lambda_{x_i} \valpha_{x_i}$, so
\[
T^{\otimes r} \valpha_\vx = \bigotimes_{i = 1}^r \lambda_{x_i} \valpha_{x_i} =
\left(\prod_{i = 1}^r \lambda_{x_i}\right) \left(\bigotimes_{i = 1}^r \valpha_{x_i} \right) =
\left(\prod_{i = 1}^r \lambda_{x_i}\right) \valpha_{\vx}.
\]

So, recalling that $T$ is symmetric and hence $T^{\otimes r}$ is also thus,
\begin{multline*}
\langle f, T^{\otimes r} f \rangle
=
\sum_{\vx \in [D]^r} \hat{f}(\vx) \widehat{T^{\otimes r} f}(\vx)
=
\sum_{\vx \in [D]^r} \hat{f}(\vx) \langle T^{\otimes r} f, \valpha_\vx\rangle \\
=
\sum_{\vx \in [D]^r} \hat{f}(\vx) \langle f, T^{\otimes r} \valpha_\vx\rangle 
=
\sum_{\vx \in [D]^r} \hat{f}(\vx) \left\langle f, \left(\prod_{i = 1}^r \lambda_{x_i} \right) \valpha_\vx\right\rangle  \\
=
\sum_{\vx \in [D]^r} \left(\prod_{i = 1}^r \lambda_{x_i} \right) \hat{f}(\vx) \left\langle f, \valpha_\vx\right\rangle 
=
\sum_{\vx \in [D]^r} \left(\prod_{i = 1}^r \lambda_{x_i} \right)
  \hat{f}^2(\vx).
\end{multline*}
Observe now that $\prod_{i = 1}^r \lambda_1 = 1$ and that $\prod_{i = 1}^r \lambda_{x_i} \geq -\rho^{|\vx|}$. Noting that $\hat{f}(\mathbf{1}) = \langle \mathbf{1}, f \rangle = \Exp_{\vx}f(\vx) = \mu$, we have that
\[
\langle f, T^{\otimes r} f \rangle \geq 2\mu^2 - \sum_{\vx \in [D]^r} \rho^{|\vx|} \hat{f}^2(\vx) = 2\mu^2 - \S_\rho(f).\qedhere
\]
\end{proof}

We will also need an estimate for $\Lambda_\rho(\mu)$. The following appears in~\cite{KKMO07}.

\begin{theorem}[{\cite[Proposition 11]{KKMO07}}]\label{thm:estimate}
    For all small enough $\mu$ and $0 < \rho \leq 1 / \ln^3(1 / \mu)$, we have
    \[
    \Lambda_\rho(\mu) \leq \mu \left(\mu + 2\rho \mu \ln (1 / \mu) \left(1 + O\left(\frac{\ln \ln (1 / \mu) + \ln \ln (1 / \rho)}{\ln(1/\mu)}\right)\right)\right).
    \]
\end{theorem}

We will also need an estimate for when $\mu$ is not small. A fact similar to the following appears in the proof of~\cite[Proposition 12]{KKMO07}; we offer a derivation from the literature for completeness.

\begin{proposition}\label{prop:bigmu}
    Suppose $\rho \geq 0$ is small and $0 < \mu < 1$. Then
    \[
    \Lambda_{\rho}(\mu) \leq \mu^2 + 3\rho.
    \]
\end{proposition}

We use the following bound due to Willink~\cite{WillinkBounds}. Their results are expressed using the following functions:
\begin{align*}
L(u, v, \rho) &= \Pr[ x \geq u, y \geq v ] \\
\Phi(u) &= \Pr[ x \leq u ] \\
\Phi(u, v, \rho) &= \Pr[ x \leq u, y\leq u],
\end{align*}
where $(x, y) \sim \N(\mathbf{0}, \mathbf{\Sigma})$, $\mathbf{\Sigma} = \begin{pmatrix}
    1 & \rho \\ \rho & 1
\end{pmatrix}$. Observe that $(-x, -y)$ and $(x, y)$ have the same distribution; furthermore $\Pr[x \geq u, y \geq v] = \Pr[-x \leq -u, -y \leq -v]$. Hence we see that if $u$ is selected so that $\Phi(u) = \Pr[x \leq u] = \mu$, then $\Lambda_\rho(\mu) = \Pr[x \leq u, y \leq u] = \Pr[-x \geq -u, -y \geq -u] = L(-u, -u, \rho)$.

\begin{theorem}[{\cite[Equation (1.2)]{WillinkBounds}}]\label{thm:wil}
Define $\theta = \sqrt{\frac{1 - \rho}{1 + \rho}}$. For $h > 0, \rho \geq 0$, we have
\[
 L(h, h, \rho) \leq \Phi(-h) \Phi(-\theta h) (1 + \rho).
\]
\end{theorem}

Also recall the following intuitive fact.

\begin{theorem}[{\cite[Equation (1.1)]{WillinkBounds}}]\label{thm:transform}
$\displaystyle L(h, k, \rho) = 1 - \Phi(h) - \Phi(k) + \Phi(h, k, \rho)$.
\end{theorem}

\begin{proof}[Proof of Proposition~\ref{prop:bigmu}]
First suppose $\mu < 1/2$. Take $u < 0$ such that $\Phi(u) = \mu$. Note that as $\Lambda_\rho(\mu) = L(-u, -u, \rho)$, we have
\[
\Lambda_\rho(\mu) \leq \Phi(u) \Phi(\theta u)(1 + \rho).
\]

Let $\phi(x) = \Phi'(x)$ be the density function of the normal distribution. It is well known that $\Phi(u) \leq -\phi(u) / u$ for $u < 0$. So, as $\max_x \phi(x) = 1 / \sqrt{2\pi} \leq 1$, we have that  $- 1 / u \geq - \phi(u) / u \geq \Phi(u) = \mu$, hence $-u \leq 1 / \mu$. Observe now that
\[
\Phi(\theta u) \leq \Phi(u) - (1 - \theta)u \leq \mu + \frac{1 - \theta}{\mu}
\]
as $\Phi$ is Lipschitz with constant $ 1/ \sqrt{2\pi} \leq 1$. So,
\[
\Lambda_\rho(\mu) \leq \Phi(u) \Phi(\theta u) (1 + \rho) \leq \mu \left( \mu + \frac{1 - \theta}{\mu}\right) (1 + \rho).
\]
Now note that $\theta = \sqrt{(1 - \rho) / (1 + \rho)} \geq 1 - \rho$ for $0 \leq \rho \leq 1$, hence
\[
\Lambda_\rho(\mu) \leq (1 + \rho)(\mu^2 + \rho)  = \mu^2 + \rho(\mu^2 + 1) + \rho^2 \leq \mu^2 + 3\rho,
\]
for small enough $\rho$.

Now we deal with the case $\mu > 1 / 2$ i.e.~$u > 0$. By Theorem~\ref{thm:transform}, and since $\Phi(-u) = 1 - \Phi(u) = 1 - \mu$, we have that
\[
\Lambda_\rho(\mu) = L(-u, -u, \rho) = 1 - 2 \Phi(-u) + \Phi(-u, -u, \rho) = 1 - 2(1 - \mu) + \Lambda_\rho(1 - \mu).
\]
Now apply the result we have proved above to $1 - \mu < 1 / 2$, to find that
\[
\Lambda_\rho(\mu) \leq 1 - 2(1 - \mu) + (1 - \mu)^2 + 3 \rho = \mu^2 + 3\rho.
\]
The case $\mu = 1 / 2$ follows by continuity.
\end{proof}


\subsection{Proof of Proposition~\ref{prop:bound}}

In this section, we prove Proposition~\ref{prop:bound}, which we restate here.

\propbound*

\noindent 
We first prove the following technical lemma.\footnote{A variant of this lemma
exists implicitly within~\cite{KKMO07}, in the proof of~\cite[Proposition
12]{KKMO07}, but this assumes that $F_T$ is convex on an interval of the form $(0, c)$ where $c$ does not depend on $T$. This seems to not be the case, so we give a different proof here. The proof in~\cite{Guruswami13:toc} claims their bound follows precisely as in~\cite{KKMO07}, so it also implicitly makes this claim about $F_T$}

\begin{lemma}\label{lem:finalPart}
    Let $F_T(x) = x^2(1 + T \ln x)$, where $F_T(0) = 0$. Fix $\ell$ and take some $T > 0$ smaller than an absolute constant, such that $\ell < e^{1 / T}$. Suppose $x_1 + \cdots + x_\ell = 1$, $x_i \geq 0$. Then, $\sum_{i = 1}^\ell F_T(x) \geq 1 / \ell - T \ln \ell / \ell - 4\ell e^{-1 / T}$.
\end{lemma}
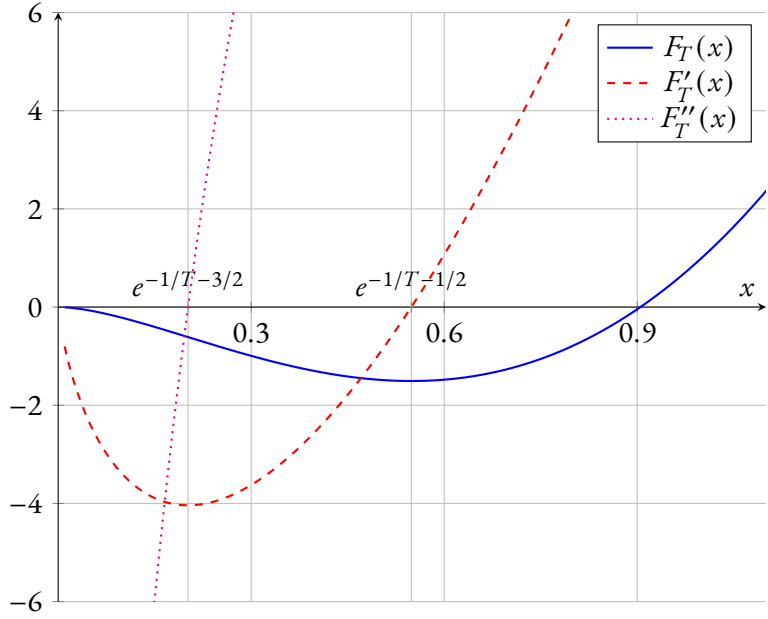
\begin{figure}
\begin{center}
\begin{tikzpicture}[>=stealth]
    \begin{axis}[
        xmin=0,xmax=1.1,
        ymin=-6,ymax=6,
        xtick distance=0.3,
        xlabel={$x$},
        grid=both,
        axis x line=middle,
        axis y line=left, 
        major tick style={grid = major},
        width=\textwidth*0.7,
        extra x ticks={0.201897,0.548812},
        extra x tick labels={$e^{-1/T - 3/2}$, $e^{-1/T - 1/2}$},
        extra x tick style={
            xticklabel style={yshift=0.5ex, anchor=south},
        }
        ]
        \addplot[no marks,blue,thick] expression[domain=0.01:2,samples=300]{x * x * (1 + 10 * ln(x))};
        \addlegendentry{$F_T(x)$}
        \addplot[no marks,red,thick,dashed] expression[domain=0.01:2,samples=300]{2 * 10 * x * ln(x) + 10 * x + 2 * x};
        \addlegendentry{$F_T'(x)$}
        \addplot[no marks,magenta,thick, dotted] expression[domain=0.01:2,samples=300]{2 * 10 * ln(x) + 3 * 10 + 2};
        \addlegendentry{$F_T''(x)$}
    \end{axis}

\end{tikzpicture}
\end{center}
\caption{Plot of $F_T(x), F_T'(x), F_T''(x)$ for $T = 10 $.}\label{fig:plots2}
\end{figure}
\begin{proof}
    Note that $F_T'(x) = 2Tx \ln x + Tx + 2x$ and $F_T''(x) = 2T \ln(x) + 3T + 2$. The second derivative is negative for $x < e^{- 1 / T - 3 / 2}$, positive for $x > e^{-1/ T - 3 / 2}$ and zero for $x = e^{-1 / T - 3 / 2}$. The first derivative is zero at $e^{-1 / T - 1 / 2}$, negative for smaller $x \geq 0$, and positive for greater $x$.
    Hence the function decreases below $e^{-1 / T - 1 / 2}$, then increases above it, and is convex whenever $x > e^{- 1 / T - 3 / 2}$. These functions can be seen for $T = 10$ in Figure~\ref{fig:plots2}.
    
    We wish first to prove that the function is Lipschitz continuous on $[0, 1]$. Observe that the derivative is minimised at $e^{-1 / T - 3/2}$, and is maximised at $1$. At $e^{-1 / T - 3 / 2}$ the derivative is
    \[
    2 T e^{-1 / T - 3 / 2} \left(- \frac{1}{T} - \frac{3}{2}\right) + T e^{- 1 / T - 3 / 2} + 2T,
    \]
    which for small $T$ is at least $-1$; furthermore at 1 the derivative is just $T + 2$, which for small $T$ is at most $3$. So we find that $F_T$ is Lipschitz continuous with constant 3 for all small enough $T$.
    
   Split $[\ell]$ into two sets $A, B$:
    let $i \in A$ if $x_i > e^{- 1 / T - 3 / 2}$, and let $i \in B$ otherwise. Since $\ell e^{-1 / T - 3 / 2} < 1$, we have that $|A| \geq 1$. Now consider two cases.
    \begin{description}
    \item[Summing over $\bm{A}$] Observe that $\sum_{i \in A} x_i = 1 - \sum_{i \in B} x_i$. Note that since $F_T$ is convex above $e^{-1 / T - 3 / 2}$, the minimum value of $\sum_{i \in A} F_T(x_i)$ when $\sum_{i \in A} x_i$ is fixed is attained when all $x_i$ for $i \in A$ are equal. Hence
    \begin{multline*}
    \sum_{i \in A} F_T(x_i) \geq 
    |A| F_T\left(\frac{1 - \sum_{i \in B}x_i}{|A|} \right)
    \geq |A| F_T(1 / |A|) - 3 \frac{\sum_{i \in B} x_i}{|A|} \\
    \geq |A| F_T(1 / |A|) - 3 \ell e^{-1 / T - 3 / 2}
    \end{multline*}
    by Lipschitz continuity, and since $x_i \leq e^{-1 / T - 3 /2 }$ when $i \in B$. Observe that this quantity is just
    \[
    \frac{1 - T \ln |A|}{|A|} - 3 \ell e^{- 1 / T - 3 / 2} \geq \frac{1 - T \ln |A|}{|A|} - 3 \ell e^{-1/T}.
    \]
    Now consider the function
    \[
    a \mapsto \frac{1 - T \ln a}{a}.
    \]
    The first derivative of this is
    \[
    \frac{T \ln a - T - 1}{a^2},
    \]
    For $a \leq \ell \leq e^{1 / T}$, we have that this derivative is negative. Thus the function is minimised when $a$ is as large as possible i.e.~$|A| = \ell$. Hence, we find that.
    \[
    \sum_{i \in A} F_T(x_i) \geq \frac{1}{\ell} - \frac{T \ln \ell}{\ell} - 3 \ell e^{-1 / T }.
    \]
    \item[Summing over $\bm{B}$] For $x_i \leq e^{-1/T - 3/2}$ we have that that $F_T(x_i)$ is decreasing, with its minimum at $e^{-1/T - 3/2}$. Note that this minimum is
    \[
    F_T(e^{-1/T-3/2}) = -\frac{3T e^{-2/T - 3}}{2} \geq -e^{-1/T},
    \]
    for small enough $T > 0$. Hence, summing over $B$, we find that
    \[
    \sum_{i \in B} F_T(x_i) \geq -\ell e^{-1 / T}.
    \]
    \end{description}
    Thus we conclude that
    \[
    \sum_{i = 1}^\ell F_T(x_i) \geq \frac{1}{\ell} - \frac{T \ln \ell}{\ell} - 4 \ell e^{- 1 / T}.\qedhere
    \]
\end{proof}

The proof of Proposition~\ref{prop:bound} given below follows the proof of~\cite[Proposition 12]{KKMO07} and of~\cite[Propositoin 3.21]{Guruswami13:toc}.

We also use the following notation from~\cite{KKMO07}: let ${[x]}^+ = \max(x, 0)$.

\begin{proof}[Proof of Proposition~\ref{prop:bound}]
  Observe that, as in the proof of Proposition~\ref{prop:boundskl}, we need only prove this result for large enough $k$. Thus assume $k$ is large. Fix $\tau, d$ to be those numbers given by Theorem~\ref{thm:modifiedMOO} for $k, \rho = c / (k - 1), \epsilon = 1 / \ell^3$.

We will actually prove that
\begin{equation}\label{eq:conclusion}
\langle f, T^{\otimes r} f \rangle
    \geq
    \frac{1}{\ell} - \frac{2 c \ln \ell}{(k-1) \ell} - C \frac{2c\ln \ell}{\ell} \frac{\ln \ln k}{(k - 1) \ln k} - 4\ell e^{-k / 3c} - D \ell e^{-\sqrt[3]{k - 1}} - \frac{1}{\ell^2}
\end{equation}
for some absolute constants $C, D > 0$.
Since $\ell \leq M(k) = o(e^{\sqrt[3]{k}})$, $2 c \ln \ell / (k - 1) \ell = 2c\ln \ell / k\ell + 2 c\ln \ell / k(k-1)\ell = 2 c\ln \ell / k\ell + o(\ln \ell / k \ell)$, and also $1 / \ell^2 \leq 1 / k\ell = o(\ln \ell / k\ell)$, our conclusion that
\[
\langle f, T^{\otimes r} f \rangle
    \geq \frac{1}{\ell} - \frac{2c \ln \ell}{k\ell} - 
    o\left(\frac{\ln \ell}{k \ell}\right).
\]
follows immediately.

Define $f^1, \ldots, f^\ell : {[k]}^r \to \Delta_\ell$ by $f(\vx) = (f^1(\vx), \ldots, f^\ell(\vx))$, and define $\mu_i = \Exp_\vx [f^i(\vx)]$. Since $\sum_i f^i(\vx) = 1$, by linearity we have that $\sum_i \mu_i = 1$.
By Lemma~\ref{lem:Sanglerelation},
\[
\langle f^i, T^{\otimes r} f^i \rangle \geq \mu_i^2- \S_{c / (k - 1)}(f^i).
\]
Furthermore, since the codomain of $f^i$ contains only nonnegative numbers we have that $\langle f^i, T^{\otimes r} f^i \rangle \geq 0$. Hence
\[
\langle f^i, T^{\otimes r} f^i \rangle \geq {\left[2\mu_i^2 - \S_{c / (k-1)}(f^i)\right]}^+,
\]
and summing over $\ell$, we have
\begin{equation}\label{eq:first}
\langle f, T^{\otimes r} f \rangle =  \sum_{i = 1}^\ell \langle f^i, T^{\otimes r} f^i \rangle \geq \sum_{i = 1}^\ell {\left[2\mu_i^2 - \S_{c / (k-1)}(f^i)\right]}^+.
\end{equation}
Our goal will now be to prove the inequality
\begin{equation}\label{eq:ineq}
\sum_{i = 1}^\ell {\left[2\mu_i^2 - \S_{c / (k-1)}(f^i)\right]}^+ \geq
\left(\sum_{i = 1}^\ell \mu_i^2 - \frac{2c\mu_i^2 \ln(1 / \mu_i)}{k - 1} \left(1 + C \frac{\ln \ln k}{\ln k}\right)\right) - D\ell e^{-\sqrt[3]{k-1}} - \frac{1}{\ell^2}
\end{equation}
for some absolute constants $C, D > 0$. We take the convention that $\mu_i^2 \ln(1 / \mu_i) = 0$ when $\mu_i = 0$.

We now split the integers $[\ell]$ into two sets $A, B$. 
Let $i \in A$ if $c / (k - 1) \leq 1 / \ln^3(1 / \mu_i)$ i.e.~$e^{-\sqrt[3]{(k - 1) / c}} \leq \mu_i$, and let $i \in B$ otherwise. We will prove~\eqref{eq:ineq} first on $A$, then on $B$, then sum. We will fix $C$ when looking at $A$, then fix $D$ depending on $C$ when looking at $B$.
\begin{description}
\item[Summing over $\bm{A}$.] We wish to prove that
\[
\sum_{i \in A} 
{\left[ 2\mu_i^2 - \S_{c / (k-1)}(f^i)\right]}^+ \geq
\left( \sum_{i \in A} \mu_i^2 - \frac{2c\mu_i^2 \ln(1 / \mu_i)}{k - 1} \left(1 + C \frac{\ln \ln k}{\ln k}\right) \right) - \frac{1}{\ell^2}
\]
for some value of $C$ that does not depend on $c, \mu_i, k$. 
First, by Theorem~\ref{thm:modifiedMOO} with $\epsilon = 1 / \ell^3$, we find that
\[
    \S_{c / (k - 1)}(f^i) \leq \Lambda_{c / (k - 1)}(\mu_i) + \frac{1}{\ell^3}\\
\]
If any $\mu_i = 1$, then the bound we want holds immediately (for large enough $k$), as then $\Lambda_{c / (k - 1)}(\mu_i) = 1$, and all other $\mu_j = 0$. Thus suppose $\mu_i < 1$. Note that by Proposition~\ref{prop:bigmu}, each term for $i \in A$ within the sum from the right-hand side of~\eqref{eq:first} contributes (for large $k$) at least $\mu_i^2 - 3c / (k - 1) - 1 / \ell^3$ to the sum. If there exists some $\mu_i > (1 / k)^{1 / 10}$, say, then these values are large enough (for large $k$) to make~\eqref{eq:conclusion} hold automatically (since for large enough $k$ we have that $(1 / k)^{1 / 5}$ is larger than $O(1 / k)$). Thus we can assume that $\mu_i \leq (1 / k)^{1 / 10}$ for all $i \in A$ i.e.~all $\mu_i$ are small for $i \in A$.

Since $c / (k - 1) \leq 1 / \ln^3(1 / \mu_i)$ and all $\mu_i$ are small when $i \in A$, we apply Theorem~\ref{thm:estimate} to find that
\begin{multline*}
    \S_{c / (k - 1)}(f^i) \leq \Lambda_{c / (k - 1)}(\mu_i) + \frac{1}{\ell^3}\\
    \leq \mu_i \left(
    \mu_i + \frac{2 c \mu_i \ln(1 / \mu_i)}{k - 1}\left(1 + O\left(\frac{\ln \ln(1 / \mu_i) + \ln \ln (k - 1) / c}{\ln(1 / \mu_i)} \right)\right)\right) + \frac{1}{\ell^3}.
\end{multline*}
We observe that as $e^{- \sqrt[3]{(k - 1) / c}} \leq \mu_i \leq (1 / k)^{1 / 10}$ for $i \in A$ there exists some constant $C$ such that for large $k$ this quantity is bounded by
\[
    \mu_i \left(
    \mu_i + \frac{2 c \mu_i \ln(1 / \mu_i)}{k - 1}\left(1 + C \frac{\ln \ln k}{\ln k} \right)\right) + \frac{1}{\ell^3}.
\]
Hence by rearranging the sum, we get that
\begin{multline*}
\sum_{i \in A} {\left[ 2\mu_i^2 - \S_{c / (k -1 )}(f^i)\right]}^+ \geq
\sum_{i \in A} \mu_i^2 - \frac{2c\mu_i^2 \ln(1 / \mu_i)}{k - 1} \left(1 + C \frac{\ln \ln k}{\ln k}\right) - \frac{1}{\ell^3} \\
\geq \left( \sum_{i \in A} \mu_i^2 - \frac{2c \mu_i^2 \ln(1 / \mu_i)}{k - 1} \left(1 + C \frac{\ln \ln k}{\ln k}\right)\right) - \frac{1}{\ell^2}.
\end{multline*}
\item[Summing over $\bm{B}$.]
We wish to prove that
\[
\sum_{i \in B} \left[2\mu_i^2 - \S_{c / (k-1)}(f^i)\right]^+ \geq
\underbrace{\left(\sum_{i \in B} \mu_i^2 - \frac{2 c \mu_i^2 \ln(1 / \mu_i)}{k - 1} \left(1 + C \frac{\ln \ln k}{\ln k}\right)\right)}_S - D\ell e^{-\sqrt[3]{k-1}}.
\]
Note that the left-hand side is nonnegative; furthermore, every term in $S$ is, for large $k$, at most some universal constant times $e^{-2\sqrt[3]{(k-1) / c}}$. By assumption $c \leq 4$, so $e^{-2 \sqrt[3]{(k-1)/c}} \leq e^{-\sqrt[3]{k - 1}}$. Thus setting $D$ large enough makes the inequality true.
\end{description}
Now, by adding the bound when summing over $A$ and $B$, we get that~\eqref{eq:ineq} is true. Combined with~\eqref{eq:ineq} with~\eqref{eq:first}, what we must now show to prove~\eqref{eq:conclusion} is:
\begin{equation}\label{eq:goal}
\sum_{i = 1}^\ell \mu_i^2 - \frac{2c \mu_i^2 \ln(1 / \mu_i)}{k - 1} \left(1 + C \frac{\ln \ln k}{\ln k}\right)
\geq
\frac{1}{\ell} - \frac{2 c \ln \ell}{(k-1) \ell} - C \frac{2 c \ln \ell}{\ell} \frac{\ln \ln k}{(k - 1) \ln k} - 4 \ell e^{- k / 3c}.
\end{equation}
Recall the function $F_T(x)$ from Lemma~\ref{lem:finalPart}.
We observe that~\eqref{eq:goal} is equivalent to 
\[
\sum_{i = 1}^\ell F_T(\mu_i) \geq
\frac{1}{\ell} - \frac{2 c \ln \ell}{(k-1) \ell} - C \frac{2 c \ln \ell}{\ell} \frac{\ln \ln k}{(k-1) \ln k},
\]
where
\[
T = \frac{2c}{k-1}\left(1 + C \frac{\ln \ln k}{\ln k}\right) \geq 0.
\]
For large $k$, we have that $T \leq 3c / k$. Hence for large $k$ we have $T$ arbitrarily small; furthermore, by assumption $\ell < M(k) = o(e^{\sqrt[3]{k}}) = o(e^{k / 3c}) = o(e^{1 / T})$. Thus, for large $k$, we have that $\ell < e^{1 / T}$. Thus, applying Lemma~\ref{lem:finalPart}, we have simply that
\[
\sum_{i = 1}^\ell F_T(\mu_i) \geq \frac{1}{\ell} - \frac{\ln \ell}{\ell}
\left(\frac{2c}{k-1}\left(1 + C \frac{\ln \ln k}{\ln k}\right) \right) - 4 \ell e^{- k / 3c},
\]
which by rearranging yields~\eqref{eq:goal}.
\end{proof}

\subsection{Proof of Proposition~\ref{prop:reduction}}

We will prove the following hardness fact.

\propreduction*

The construction here is very standard, and essentially identical to that in~\cite[Section 11.4]{KKMO07} or~\cite[Section 3.4]{Guruswami13:toc}. We will express our results in a more algebraic way, though, rather than using the language of PCP verifiers.

\begin{proof}[Proof of Proposition~\ref{prop:reduction}]

Fix $p \in \{1, 2\}$. Henceforth let $\overline{\cdot}, \underline{\cdot}$ be $\overline{\cdot}^{(p)}$ and $\underline{\cdot}_{(p)}$ respectively.
We will consider some value $\eta$ that depends on $p, \epsilon$; at the end of the
proof we will fix $\eta$ small enough for everything to follow. By assumption,
there exists an $r \in \mathbb{N}$ such that, given a
left-regular label cover instance $I = (V = \Va \cup \Vb, E, \pi_*)$, with
$p$-to-1 constraints,
it is \NP-hard to decide whether there exists a solution with value at least $1 -
\eta$, or all solutions have value at most $\eta$.

We observe that a constraint is a $p$-to-1 function from $[pr]$ to $[r]$; such a function can be written as a composite between a permutation $\pi : [pr] \to [pr]$ and the function $\sigma : [pr] \to [r]$, given by $\sigma(1) = \ldots = \sigma(p) = 1, \sigma(p + 1) =  \ldots = \sigma(2p) = 2$  i.e.~$\sigma(x) = \lceil x / p \rceil$. Thus we assume that the constraint that corresponds to the edge $(a, b)$ is given by $\sigma \circ \pi_{a, b}$.

We will reduce
this instance in polynomial time to an instance of maximum $k$- vs.
$\ell$-colouring, namely a graph $G = (V', E')$, and then prove the completeness
  and soundness of the reduction. 

\paragraph*{Reduction} 
For every variable $b \in \Vb$, introduce a set of variables in our graph in the following way. For every vector $\vx \in {[k]}^{pr}$, we introduce a vertex $v_b(\vx)$. Thus our graph $G$ will have the vertex set
\[
V' = \{ v_b(\vx) \mid b \in \Vb, \vx \in {[k]}^{pr} \},
\]
with $k^{pr} |\Vb|$ vertices. As for the edges, consider every pair of edges $(a,
  b), (a, b') \in E$. For such a pair, for every $\vx, \vy \in [k]^{pr}$, we add in an edge between $(\vx^{\pi_{a, b}})$ and $v_{b'}(\vx^{\pi_{a, b'}})$ exactly
  $L^{r} T^{\otimes r}(\overline{\vx} \leftrightarrow \overline{\vy}) \leq L^r$ times. This is well defined since $\overline{\vx}, \overline{\vy} \in [k^p]^r$, and $T^{\otimes r}$ can be seen as a Markov chain over $[k^p]^r$. Furthermore each probability in $T$ is a nonnegative integer multiple of $L$, hence we never add non-integer or negative numbers of edges between vertices.
  
Throughout the following, we will let $f_a(\vx)$ denote the colour of $v_a(\vx)$. Observe that this reduction works in polynomial time, since $L, k, p, r$ are constants.

\paragraph*{Completeness}
Suppose there exists a solution to $I$, say $c : V \to [r]$, that satisfies a $1 - \eta$ fraction of constraints. Consider the following $k$-colouring of $G$: let $f_v(x_1, \ldots, x_r) = x_{c(v)}$. Now, we must compute the value of this colouring. Suppose $\val(a)$ is the proportion of edges incident to $a \in \Va$ solved by $c$. Since the instance $I$ is regular on $\Va$, we have that the value of the instance is $\Exp[ \val(a)] \geq 1 - \eta$, where $a$ is drawn uniformly at random from $\Va$.

What fraction of the edges in $E'$ were constructed due to a pair of edges $(a,
  b), (a, b') \in E$ which are both satisfied by $c$? Since the degree of all
  vertices in $\Va$ is equal, and thus the same number of edges is added for
  each $a$, this is equivalent to asking ``what is the probability that, if we
  select $a \in \Va$ and $b, b' \in \Vb$ incident to $a$ uniformly and
  independently at random, then the edges $(a, b), (a, b')$ are solved by $c$''.
  But observe that this probability, for fixed $a$, is
  at least $1 - 2(1 - \val(a)) = 2\val(a) - 1$. By linearity of expectation, the required probability is thus at least $\Exp[2 \val(a) - 1] \geq 2(1 - \eta) - 1 = 1 - 2\eta$.

Now, note that every edge $(v_b(\vx^{\pi_{a, b}}), v_{b'}(\vy^{\pi_{a, b'}})) \in E'$ where $(a, b), (a, b')$ are solved by $c$ will also be properly coloured by $c'$. To see why, note that $f_b(\vx^{\pi_{a, b}}) = f_b(x_{\pi_{a, b}(1)}, \ldots, x_{\pi_{a, b}(pr)}) = x_{\pi_{a, b}(c(b))}$, and likewise $f_{b'}(\vy^{\pi_{a, b'}}) = y_{\pi_{a, b'}(c(b'))}$. Defining $i = \pi_{a, b}(c(b)), j = \pi_{a, b'}(c(b'))$, these are $x_i$ and $y_j$. But since the edges $(a, b), (a, b')$ are solved, we have  $\sigma(i) = c(a) = \sigma(j)$,
or equivalently
\[
i, j \in \{ p c(a), p c(a) + 1, \ldots, p c(a) + p - 1 \}.
\]

Since we have added in the edge $(v_b(\vx^{\pi_{a, b}}), v_{b'}(\vy^{\pi_{a, b'}}))$, the transition between $\overline{\vx}$ and $\overline{\vy}$ in $T^{\otimes r}$ has nonzero probability. Thus, the transition between $(x_{p c(a)}, \ldots, x_{p c(a) + p - 1})$ and $(y_{p c(a)}, \ldots, y_{p c(a) + p - 1})$ has nonzero probability in $T$; by the colourfulness of $T$,
\[
\{ x_{pc(a)}, \ldots, x_{pc(a) + p - 1} \} \cap \{ y_{pc(a)}, \ldots, y_{pc(a) + p - 1} \} = \emptyset.
\]
Thus $x_i \neq y_i$. Hence the edge is properly coloured, as
\[
f_b(\vx^{\pi_{a, b}}) 
=
x_{\pi_{a, b}(c(b))}
= x_{i}
\neq y_j
= y_{\pi_{a, b'}(c(b'))}
= f_{b'}(\vy^{\pi_{a, b'}}).
\]
Thus the resulting graph has a $k$-colouring of value $1 - 2\eta$; taking $2\eta < \epsilon$ thus implies completeness.

\paragraph*{Soundness}
Suppose that the graph $G$ has an $\ell$-colouring of value at least $\beta +
  \epsilon$; call it $c$. As opposed to the completeness case, let $\val(a)$
  denote the proportion of edges in $G$ added due to edges $(a, b), (a, b') \in
  E$ that are properly coloured. Having fixed such an $a$, note that every choice of $\vx \in [k^p]^r$ induces the same number of edges, and for this $\vx$ every choice of $\vy$ induces a number of edges proportional to $T^{\otimes}(\vx \leftrightarrow \vy)$. Hence
  \[
  \val(a) = \Exp_{\substack{\overline{\vx} \in [k^p]^r \\ \overline{\vy} \sim T^{\otimes r}(\overline{\vx})}} [ f_a(\vx^{\pi_{a, b}}) \neq f_b(\vy^{\pi_{a, b}}) ].
  \]

  Observe that since all vertices in $\Va$ have
  the same degree, similarly the same number of edges of the form $(a, b), (a,
  b') \in E$ exist for every $a \in \Va$; hence the value of $c$ can be
  expressed as $\Exp[ \val(a) ] \geq \beta + \epsilon$, where $a$ is drawn
  uniformly at random from $\Va$. Now apply Markov's inequality to $1 - \val(a)$
  to find that the probability that $1 - \val(a) \geq 1 - \beta$,
  i.e., $\val(a) \leq \beta$ is at most $(1 - \beta - \epsilon) / (1 - \beta) = 1 - \epsilon / (1 - \beta) \leq 1 - \epsilon$. Hence the probability that $\val(a) > \beta$ is at least $\epsilon$. Let $S \subseteq \Va$ be the set of $a \in \Va$ for which $\val(a) \geq \beta$; we have that $|S| \geq \epsilon |\Va|$.

Our goal will be to assign, to each vertex in $S \cup \Vb$, at most $C = \max (\lceil 2pd / \tau \rceil, 1)$ possible values; if $c(a) \subseteq [r]$ for $a \in S$ and $c(b) \subseteq [pr]$ for $b \in \Vb$ is the set of values, we will want
\begin{equation}\label{eq:intersections}
    c(a) \cap \sigma(\pi_{a, b}(c(b))) \neq \emptyset
\end{equation}
for at least a $\tau / 2p$ fraction of the edges $(a, b) \in E$ \emph{for any fixed $a \in S$}. Then, by randomly selecting a value for $v$ from $c(v)$, we find that each of these $\tau / 2p$ fraction of edges of the form $(a, b) \in E$ for any fixed $a \in S$ are satisfied with probability at least $1 / C^2$. Since $I$ is regular on $\Va$ and $|S| \geq \epsilon |\Va|$, this implies that this solution (if extended arbitrarily to $\Va \setminus S$) has value at least $\epsilon \tau / 2pC^2$.
Taking $\eta$ small enough so that $\epsilon \tau / 2pC^2 \geq \eta$ is then enough to prove soundness.

Note that $[ f_b(\vx) \neq f_{b'}(\vy) ] = 1 - f_b(\vx)\cdot f_{b'}(\vy)$, if we see the codomain $[\ell]$ of $f_b$ as being embedded within $\Delta_\ell$; fixing some $a \in S$, we then observe that
\begin{multline*}
\beta < \val(a) = \Exp_{\substack{(a, b), (a, b') \in E\\ \overline{\vx} \in {[k^p]}^{r} \\ \overline{\vy} \sim T^{\otimes r}(\overline{\vx})}} \left[ 1 - f_b(\vx^{\pi_{a, b}}) \cdot f_{b'}(\vy^{\pi_{a, b'}})\right] \\
=
1 - \Exp_{\substack{\overline{\vx} \in [k^p]^r \\ \overline{\vy} \sim T^{\otimes r}(\overline{\vx})}} \left[
\Exp_{(a, b), (a, b') \in E} 
\left[ f_b(\vx^{\pi_{a, b}}) \cdot  f_{b'}(\vy^{\pi_{a, b'}}) \right]\right].
\end{multline*}
Observe that the inner product above is between two independent variables, hence the expression is equal to
\[
1 - \Exp_{\substack{\overline{\vx} \in [k^p]^r \\ \overline{\vy} \sim T^{\otimes r}(\overline{\vx})}} \left[
{\Exp_{(a, b) \in E} 
\left[f_b(\vx^{\pi_{a, b}})\right]} \cdot \Exp_{(a, b) \in E} \left[ f_{b}(\vy^{\pi_{a, b}})\right]\right].
\]
Now, define $g_a(\vx) = \Exp_{(a, b) \in E} \left[f_b(\vx^{\pi_{a, b}}) \right]$ i.e.~$g_a = \Exp_{(a, b) \in E} [ f_b^{\pi_{a, b}} ]$. Hence, by substituting $\vx, \vy$ with $\underline{\vx}, \underline{\vy}$, we get that the expression is
\[
1 - \Exp_{\substack{\overline{\vx} \in [k^p]^r \\ \overline{\vy} \sim T^{\otimes r}(\overline{\vx})}} \left[
g_a(\vx) \cdot g_a(\vy) \right]
=
1 - \Exp_{\substack{\vx \in [k^p]^r \\ \vy \sim T^{\otimes r}(\vx)}} \left[
\overline{g_a}(\vx) \cdot \overline{g_a}(\vy) \right].
\]
Now, separating the choice of $\vx$ from $\vy$, and by linearity, this is
\[
1 - \Exp_{\vx \in [k^p]^r} \left[
\overline{g_a}(\vx) \cdot \Exp_{\vy \sim T^{\otimes r}(\vx)} \overline{g_a}(\vy) \right]
=
1 - \Exp_{\vx \in [k^p]^r} \left[
 \overline{g_a}(\vx)\cdot \left((T^{\otimes r} \overline{g_a})(\vx) \right ) \right]
=
1 - \langle \overline{g_a},T^{\otimes r} \overline{g_a} \rangle.
\]

Hence, $\langle \overline{g_a}, T^{\otimes r} \overline{g_a} \rangle < 1 -
\beta$ for $a \in S$. By assumption we must have $\Inf_i^{\leq
d}(\overline{g_a}) > \tau$ for at least one $i \in [r]$. By
Corollary~\ref{corr:relabelling}, for at least one $i \in [pr]$ we must have
$\Inf_i^{\leq d}(g_a) > \tau / p$. Let $c(a) = \{ i \}$; thus $|c(a)| \leq C$. Thus we have labelled $S$. For $b \in \Vb$, define
\[
c(b) = \{i \in [r] \mid \Inf_i^{\leq d} (f_b) \geq \tau / 2p\}.
\]
Note that $\sum_{i = 1}^r \Inf_i^{\leq d}(f_b) \leq d$, and furthermore $\Inf_i^{\leq d}(f_b)\geq 0$ as $f_b$ takes nonnegative values, so $|c(b)| \leq 2pd / \tau\leq C$.

Now, we must prove that this mapping $c$ satisfies~\eqref{eq:intersections} for
any fixed $a \in S$, and for a $\tau / 2p$ fraction of edges of the form $(a, b) \in E$. Fixing $a \in S$, letting $g = g_a$, $\{ i \} = c(a)$, we observe first that
\[
\hat{g} = \widehat{\Exp_{(a, b) \in E}[ f_b^{\pi_{a, b}} ]} = \Exp_{(a, b) \in E} [ \widehat{f_b^{\pi_{a, b}}} ],
\]
by linearity of $\hat{\cdot}$ and hence $\hat{g}(\vx) = \Exp_{(a, b) \in E} [ \widehat{f_b^{\pi_{a, b}}}(\vx)]$. Thus
\begin{multline*}
\tau/p
< \Inf_{i}^{\leq d}(g)
= \sum_{\substack{\vx \in [k]^r \\ x_i \neq 1 \\ |\vx| \leq d}} \left|\hat{g}(\vx)\right|^2
= \sum_{\substack{\vx \in [k]^r \\ x_i \neq 1 \\ |\vx| \leq d}} \left|\Exp_{(a, b) \in E} [\widehat{ f_b^{\pi_{a, b}}}(\vx)] \right|^2
\leq
\sum_{\substack{\vx \in [k]^r \\ x_i \neq 1 \\ |\vx| \leq d}} \Exp_{(a, b) \in E} \left[\left|\widehat{ f_b^{\pi_{a, b}}}(\vx) \right|^2\right] \\
= \Exp_{(a, b) \in E} \left[ \sum_{\substack{\vx \in [k]^r \\ x_i \neq 1 \\ |\vx| \leq d}} \left|\widehat{ f_b^{\pi_{a, b}}}(\vx) \right|^2\right]
= \Exp_{(a, b) \in E}\left[\Inf_{i}^{\leq d}(f_b^{\pi_{a, b}})\right],
\end{multline*}
where all the inequalities follow by linearity or convexity. We now apply
Markov's inequality to $\max(\Inf_i^{\leq d}(f_b^{\pi_{a, b}})) - \Inf_i^{\leq
d}(f_b^{\pi_{a, b}})$. For a $\tau / 2p$ fraction of the edges $(a, b) \in E$ we have
\[
\Inf_i^{\leq d}(f_b^{\pi_{a, b}}) \geq \tau / 2p.
\]
Now note that for such $b$,
\begin{multline*}
\tau / 2p \leq \Inf_i^{\leq d}(f_b^{\pi_{a, b}})
=
\sum_{\substack{\vx \in [k]^r \\ x_i \neq 1 \\ |\vx| \leq d}} \left|\widehat{ f_b^{\pi_{a, b}}}(\vx) \right|^2\
=
\sum_{\substack{\vx \in [k]^r \\ x_i \neq 1 \\ |\vx| \leq d}} \left|\hat{f_b}({\vx^{\pi_{a, b}}}) \right|^2 \\
=
\sum_{\substack{\vy \in [k]^r \\ y_{\pi^{-1}_{a, b}(i)} \neq 1 \\ |\vy| \leq d}} \left|\hat{f_b}({\vy}) \right|^2
= \Inf_{\pi^{-1}_{a, b}(i)}(f_b).
\end{multline*}
where $\vy = \vx^{\pi_{a, b}}$. Hence $\pi^{-1}_{a, b}(i) \in c(b)$, so in particular~\eqref{eq:intersections} holds for this $a$ and for at least a $\tau / 2p$ fraction of the edges $(a, b)$. This concludes the proof.
\end{proof}

\section{AGC-hardness of 1-approximation}\label{app:alpha1}

We show a reduction from AGC to $1$-approximation of $\rho_k(G)$ via $\rho_\ell(G)$. 

\alphaone*
\begin{proof}
Let $\rho = p / q$, $p > 0$. Suppose we are given a graph $G$; we are then asked to decide if it is $k$-colourable or not even $\ell$-colourable. Let $G$ have $m$ edges and let $1$ denote the graph with one vertex and an edge from that vertex to itself. (This notation is justified, since this graph is a unit with respect to the direct product of graphs.) Let $+$ denote disjoint union of graphs. We also allow multiplication of a graph by a scalar in the obvious way. (For example, $3G = G + G + G$.) Then our reduction takes the graph $G$ to the graph $pG + (q - p)m1$.

Note first that the reduction can be done in logarithmic space. For completeness, note that if $G$ is $k$-colourable, then $pG + (q-p)m1$ has a $k$-colouring of value $\rho$, namely the one that colours each of the $p$ disjoint copies of $G$ as in the $k$-colouring of $G$. This colouring correctly colours $pm$ of the $pm + (q-p)m = qm$ edges i.e.~it has value $\rho = pm/qm = p/q$. For soundness, suppose that $pG + (q-p)m1$ has an $\ell$-colouring of value $\rho$. This colouring must correctly colour a $p/q$ fraction of the edges of $pG + (q-p)m1$. Since this graph has $qm$ edges, it must correctly colour $pm$ edges. But the only edges that could possibly be correctly coloured are the ones in $pG$ (since the remaining edges in $(q-p)m1$ are all loops). Furthermore, there are $pm$ edges in $pG$, thus all the edges in $pG$ must be correctly coloured. But this implies that $G$ has an $\ell$-colouring, as required. 
\end{proof}

\bibliographystyle{plainurl}
\bibliography{nz}

\end{document}